\documentclass[a4paper,english]{lipics}

\usepackage{amsfonts,amsmath,amssymb,amsthm}
\usepackage{xcolor}
\usepackage{enumerate}
\usepackage{xspace}
\usepackage{ifthen}
\usepackage{pgf,tikz}
\usetikzlibrary{arrows,automata,shapes}
\tikzstyle{every state}=[minimum size=12pt,inner sep=0pt]
\usetikzlibrary{snakes}
\tikzstyle{randomPath}=[decorate,decoration={amplitude=1pt,segment length=2pt,random steps},very thick]
\tikzstyle{randomPath2}=[decorate,decoration={amplitude=2pt,segment length=6pt,random steps},very thick]
\usetikzlibrary{calc}

\newcommand{\FR}{\(\mathtt{FR}\)}
\newcommand{\SK}{\(\mathtt{automgrp}\)}
\newcommand{\GAP}{\(\mathtt{GAP}\)}

\newcommand{\N}{{\mathbb N}}

\newcommand{\mz}{\mathfrak m}
\newcommand{\dz}{\mathfrak d}
\newcommand{\aut}[1]{{\mathcal #1}}
\newcommand{\dual}[1]{{\mathfrak d}({#1})}
\newcommand{\mot}[1]{{\mathbf {#1}}}

\newcommand{\pres}[1]{\langle{#1}\rangle}
\newcommand{\presm}[1]{\pres{{#1}}_{+}}
\newcommand{\cd}[1]{\curlywedge({\aut{#1}})}
\newcommand{\cdv}[1]{{\Lambda}({\aut{#1}})}

\newcommand{\redEdge}{e_2}
\newcommand{\two}[1]{{\cal O}_2}

\newcommand{\rest}{\mathfrak s}

\newcommand{\otree}[1][]{\mathfrak{t}{\ifthenelse{\equal{#1}{}}{}{(\aut{#1})}}}
\newcommand{\htree}{T_{{\cal O}_2}}

\newcommand{\vhhtree}[1][v]{{\htree}_{|{\mot{#1}}}}

\newcommand{\fact}{\text{fact}}
\newcommand{\cw}{\curlywedge}

\usepackage{xcolor}

\theoremstyle{plain}
\newtheorem{proposition}[theorem]{Proposition}
\theoremstyle{remark}
\newtheorem{nremark}[theorem]{Remark}

\title{A connected 3-state reversible Mealy automaton cannot generate an infinite Burnside group
\footnote{This work was partially supported by the french \emph{Agence Nationale pour la~Recherche},
through the Project MealyM ANR-JCJC-12-JS02-012-01.
The third author was partially supported by the New Researcher Grant from USF Internal Awards Program.}}
\titlerunning{A connected 3-state reversible automaton cannot generate an infinite Burnside group}
\author[1]{Ines Klimann}
\author[1]{Matthieu Picantin}
\author[2]{Dmytro Savchuk}
\affil[1]{Univ Paris Diderot, Sorbonne Paris Cit\'e, LIAFA,
    UMR 7089 CNRS,\\ F-75013 Paris,
    France\hfill \texttt{\{klimann,picantin\}@liafa.univ-paris-diderot.fr}}
\affil[2]{Department of Mathematics and Statistics,
University of South Florida,\\
4202 E Fowler Ave --
Tampa, FL 33620-5700, USA
\hfill\texttt{savchuk@usf.edu}}
\Copyright[nd]{Ines Klimann, Matthieu Picantin, and Dmytro Savchuk}
\serieslogo{}
\volumeinfo
  {Billy Editor, Bill Editors}
  {2}
  {Conference title on which this volume is based on}
  {1}
  {1}
  {1}
\EventShortName{}
\DOI{10.4230/LIPIcs.xxx.yyy.p}

\begin{document}

\maketitle

\begin{abstract}

The class of automaton groups is a rich source of the simplest
examples of infinite Burnside groups.
However, there are some classes of automata that do not contain such examples.
For instance, all infinite Burnside automaton groups in
the literature are generated by non reversible Mealy automata and
it was recently shown that 2-state invertible-reversible
Mealy automata cannot generate infinite Burnside groups. Here we extend this
result to connected 3-state invertible-reversible Mealy automata, using new
original techniques. The results provide the first uniform method to construct
elements of infinite order in each infinite group in this class.
\end{abstract}

\section{Mealy automata and the General Burnside problem}
In 1902, Burnside has introduced a question which would become
highly influential in group theory~\cite{burnside}:
\begin{center}
\emph{Is a finitely generated group whose all elements have finite
  order necessarily finite?}
\end{center}

This problem is now known as the \emph{General Burnside Problem}. A
group is commonly called a \emph{Burnside} group if it is finitely
generated and all its elements have finite order.

In 1964, Golod and Shafarevich~\cite{golod,golod_shafarevich} were the
first ones to give a negative answer to the general Burnside problem
and around the same time Glushkov suggested that groups generated by
automata could serve as a different source of
counterexamples~\cite{glushkov:automata}. In 1972, Aleshin gave an
answer as a subgroup of an automaton group~\cite{aleshin}, and then in
1980, Grigorchuk exhibited the first and the simplest by now example
of an infinite Burnside automaton group~\cite{grigorchuk1}. Since then many
infinite Burnside automaton groups have been
constructed~\cite{bartholdi_s:growth,grigorch:degrees,gupta_s:burnside,sushch:burnside}. Even
by now, the simplest examples of infinite Burnside groups are still automaton
groups.

All the examples of infinite Burnside automaton groups in the literature
happen to be generated by non-reversible invertible Mealy automata, that is,
invertible Mealy automata where all the letters do not act as
permutations on the stateset.

It was proved in~\cite{Kli13} that a 2-state invertible-reversible
Mealy automaton cannot generate an infinite Burnside group, but the techniques
were strongly based on the fact that the stateset has size~2.
Here we address this problem for a larger class, namely the 3-state
invertible-reversible automata, and prove the following theorem.

\begin{theorem}\label{thm-main}
A connected 3-state invertible-reversible Mealy automaton cannot
generate an infinite Burnside group.
\end{theorem}

For the proof of this theorem we develop new techniques, centered on the orbit tree of the
dual of the Mealy automaton. We hope 
that these techniques could be further extended to attack similar problem for automata with bigger statesets.

The class of groups generated by automata is very interesting from algorithmic point of view. Even thought the word problem (given a word in generators decide if it represents the trivial element of the group) is decidable, most of other basic algorithmic questions, including finiteness problem, order problem, conjugacy problem, are either known to be undecidable, or their decidability is unknown. For example, it was proved recently that the order problem is undecidable for automaton semigroups~\cite{gillibert:finiteness14} and for groups generated by, so called, asynchronous Mealy automata~\cite{belk_b:undecidability}. However, this problem still remains open for the class of all groups generated by Mealy automata. There are many partial methods to find elements of infinite order in such groups, but the class of reversible automata is known as the class for which most of these algorithms do not work or perform poorly. The proof of Theorem~\ref{thm-main} gives a uniform algorithm to produce many elements of infinite order in infinite groups generated by 3-state invertible-reversible automata. Unfortunately, it does not provide an algorithm that can determine if the group itself is infinite, however, it is known that each invertible-reversible but not bireversible automaton generates an infinite group~\cite{AKLMP12}.

The paper is organized as follows. In Section~\ref{sec-basic}, we set up notation, provide well-known definitions and
facts concerning automaton groups and rooted trees. Certain results concerning
connected components of reversible Mealy automata are given in
Section~\ref{sec-cc}. In Section~\ref{sec-otree} we introduce
a crucial construction for our proofs: the labeled orbit tree of a Mealy
automaton. Finally, Section~\ref{sec-main_result} contains the proof of our main results, including Theorem~\ref{thm-main}.

\section{Basic notions}\label{sec-basic}

\subsection{Groups generated by Mealy automata}
We first recall the formal definition of an automaton. A {\em (finite, deterministic, and complete) automaton} is a
triple
\(
\bigl( Q,\Sigma,\delta = (\delta_i\colon Q\rightarrow Q )_{i\in \Sigma} \bigr)
\),
where the \emph{stateset}~$Q$
and the \emph{alphabet}~$\Sigma$ are non-empty finite sets, and
where the~\(\delta_i\)
are functions.

\smallskip

A \emph{Mealy automaton} is a quadruple
\(\bigl( Q, \Sigma, \delta = (\delta_i\colon Q\rightarrow Q )_{i\in \Sigma},
\rho = (\rho_x\colon \Sigma\rightarrow \Sigma  )_{x\in Q} \bigr)\), 
such that both~\((Q,\Sigma,\delta)\) and~\((\Sigma,Q,\rho)\) are
automata.
In other terms, a Mealy automaton is a complete, deterministic,
letter-to-letter transducer with the same input and output alphabet.

The graphical representation of a Mealy automaton is
standard, see Figures~\ref{fig-conn} and~\ref{fig-disc}.

\begin{figure}[ht]%
\centering
\begin{minipage}{.45\textwidth}%
	\centering
	\begin{tikzpicture}[->,>=latex,node distance=12mm]
	\node[state] (1) {\(x\)};
	\node (0) [right of=1] {};
	\node[state] (2) [right of=0, above of=0] {\(y\)};
	\node[state] (3) [right of=0, below of=0] {\(z\)};
	\path 
      (1) edge[loop left] node[left]{\(\begin{array}{c} 1|2\\2|1\\3|4\end{array}\)} (1)
      (1) edge node[below]{\(4|3\)} (3)
      (3) edge[bend right] node[right]{\(\begin{array}{c} 4|1\\2|3\end{array}\)} (2)
      (2) edge[bend right] node[left]{\(2|1\)} (3)
      (2) edge node[above]{\(4|3\)} (1)
      (2) edge[loop above] node[right]{\(\begin{array}{c} 1|2\\3|4\end{array}\)} (2)
      (3) edge[loop below] node[right]{\(\begin{array}{c} 1|2\\3|4\end{array}\)} (3);
	\end{tikzpicture}
\caption{An example of a 4-letter 3-state\\ connected Mealy automaton.}%
\label{fig-conn}
\end{minipage}%
\hspace*{13mm}
\begin{minipage}{.45\textwidth}%
	\centering
	\begin{tikzpicture}[->,>=latex,node distance=12mm]
	\node[state] (1) {\(x\)};
	\node (0) [right of=1] {};
	\node[state] (2) [right of=0, above of=0] {\(y\)};
	\node[state] (3) [right of=0, below of=0] {\(z\)};
	\path (1) edge[loop left] node[left]{\(\begin{array}{c} 1|1\\2|3\\3|2\end{array}\)} (1)
      (2) edge[bend right] node[left]{\(3|3\)} (3)
      (3) edge[bend right] node[right]{\(3|3\)} (2)
      (2) edge[loop above] node[right]{\(\begin{array}{c} 1|1\\2|2\end{array}\)} (2)
      (3) edge[loop below] node[right]{\(\begin{array}{c} 1|2\\2|1\end{array}\)} (3);
	\end{tikzpicture}
\caption{An example of a 3-letter 3-state\\ disconnected Mealy automaton.}%
\label{fig-disc}%
\end{minipage}%
\end{figure}

A Mealy automaton \((Q,\Sigma,\delta, \rho)\) is
\emph{invertible\/} if the functions \(\rho_x\) are permutations of~\(\Sigma\)
and \emph{reversible\/} if the functions \(\delta_i\) are
permutations of~\(Q\).

In a Mealy automaton~\(\aut{A}=(Q,\Sigma, \delta, \rho)\), the sets~\(Q\)
and~\(\Sigma\) play dual roles. So we may consider the \emph{dual (Mealy)
automaton} defined by
\(
\dual{\aut{A}} = (\Sigma,Q, \rho, \delta)
\).
Obviously, a Mealy automaton is reversible if and only if its dual is
invertible.

\medskip

Let~\(\aut{A} = (Q,\Sigma, \delta,\rho)\) be a Mealy automaton.
We view~\(\aut{A}\) as an automaton with an input and an output tape, thus
defining mappings from input words over~$\Sigma$ to output words
over~$\Sigma$.
Formally, for~\(x\in Q\), the map~$\rho_x\colon\Sigma^* \rightarrow \Sigma^*$,
extending~$\rho_x\colon\Sigma \rightarrow \Sigma$, is defined recursively by:
\begin{equation*}
\forall i \in \Sigma, \ \forall \mot{s} \in \Sigma^*, \qquad
\rho_x(i\mot{s}) = \rho_x(i)\rho_{\delta_i(x)}(\mot{s}) \:.
\end{equation*}
By convention, the image of the empty word is itself.
The mapping~\(\rho_x\) for each $x\in Q$ is length-preserving and prefix-preserving.
We say that~\(\rho_x\) is the \emph{production
function\/} associated with~\((\aut{A},x)\).
For~$\mot{x}=x_1\cdots x_n \in Q^n$ with~$n>0$, set
\(\rho_\mot{x}\colon\Sigma^* \rightarrow \Sigma^*, \rho_\mot{x} = \rho_{x_n}
\circ \cdots \circ \rho_{x_1} \:\).

Denote dually by~\(\delta_i\colon Q^*\rightarrow Q^*,
i\in \Sigma\), the production functions associated with
the dual automaton
$\dz(\aut{A})$. For~$\mot{s}=s_1\cdots s_n
\in \Sigma^n$ with~$n>0$, set~\(\delta_\mot{s}\colon Q^* \rightarrow Q^*,
\ \delta_\mot{s} = \delta_{s_n}\circ \cdots \circ \delta_{s_1}\).

\smallskip

The semigroup of mappings from~$\Sigma^*$ to~$\Sigma^*$ generated by
$\{\rho_x, x\in Q\}$ is called the \emph{semigroup generated
  by~$\aut{A}$} and is denoted by~$\presm{\aut{A}}$.
When~\(\aut{A}\) is invertible,
its production functions are
permutations on words of the same length and thus we may consider
the group of mappings from~$\Sigma^*$ to~$\Sigma^*$ generated by
$\{\rho_x, x\in Q\}$. This group is called the \emph{group generated
  by~$\aut{A}$} and is denoted by~$\pres{\aut{A}}$.

Let us recall some known results that will be used in our proofs.

\begin{proposition}[see, for example,~\cite{AKLMP12}]
\label{prop:finite_group}
An invertible Mealy automaton generates a finite group if and only if
it generates a finite semigroup.
\end{proposition}

\begin{proposition}[\cite{AKLMP12,nek,sv11}]
\label{prop:finite_dual}
A Mealy automaton generates a finite (semi)group if and only if so does
its dual.
\end{proposition}

\medskip

As in the case of automata, there is a notion of minimization for Mealy
automata: if several states have the same action, it consists of
keeping only one of them.
More formally, for a Mealy automaton
\(\aut{A}=(Q,\Sigma,\delta,\rho)\), the \emph{Nerode equivalence
  \(\equiv\) on \(Q\)} is the limit of the  sequence of increasingly
finer equivalences~$(\equiv_k)$ recursively defined by:
\begin{align*}
\forall x,y\in Q\colon \qquad\qquad x\equiv_0 y & \ \Longleftrightarrow
\ \forall i\in\Sigma\colon \rho_x(i)=\rho_y(i)\:,\\
\forall k\geqslant 0\colon x\equiv_{k+1} y &
\ \Longleftrightarrow\  \bigl(x\equiv_k y\quad \wedge\quad\forall
i\in\Sigma\colon \delta_i(x)\equiv_k\delta_i(y)\bigr)\:.
\end{align*}

Since the set $Q$ is finite, this sequence is ultimately constant.
The \emph{minimization} of~\(\aut{A}\) is then the Mealy automaton
\(\mz(\aut{A})=(Q/\mathord{\equiv},\Sigma,\tilde{\delta},\tilde{\rho})\),
where for every $(x,i)$ in $Q\times \Sigma$,
we have~$\tilde{\delta}_i([x])=[\delta_i(x)]$ and~$\tilde{\rho}_{[x]}=\rho_x$.

If the minimizations of two invertible Mealy automata are structurally
isomorphic, these automata generate two isomorphic groups. Moreover, the following lemma is straightforward:

\begin{lemma}\label{lm:powers}
Let \(\aut{A}\) be an invertible Mealy automaton. If there exist
\(p<q\) such that~\(\mz(\aut{A}^p)\) and~\(\mz(\aut{A}^q)\) are
isomorphic, then \(\aut{A}\) generates a finite group.
\end{lemma}

\subsection{Terminology on trees}\label{sec-trees}
Throughout this paper, we will use different sorts of labeled
trees. Here we set up some terminology that are common for all of them.

All our trees are rooted, \emph{i.e.} with a selected vertex called the \emph{root}. We will visualize the trees traditionally as growing down from the root. Hence, the ``top'' and ``bottom'' directions in the tree are defined as ``to'' and ``from'' the root respectively.
It will be convenient to orient all edges in the tree in the direction from the root.
The initial vertex of an edge~\(e\) is denoted by~\(\top(e)\) and its terminal vertex by~\(\bot(e)\).
A \emph{path} is a (possibly infinite) sequence of adjacent edges without backtracking
from top to bottom. A path is said to be~\emph{initial} if it starts at the root of the tree.
A \emph{branch} is an infinite initial path.
The initial vertex of a non-empty path \(\mot{e}\) is denoted by~\(\top(\mot{e})\)
and its terminal vertex by~\(\bot(\mot{e})\)

The \emph{level of a vertex} is its distance to the root and the
\emph{level of an edge} or \emph{a path} is the level of its initial vertex.
For~\(V\) a vertex in a tree~\(\aut{T}\), by \emph{section of~\(\aut{T}\)
at~\(V\)} (denoted by~\(\aut{T}_{|V}\)) we mean the subtree of~\(\aut{T}\)
with root~\(V\) consisting of all those vertices of~\(\aut{T}\) that are descendant of~$V$.
Additionally, for~\(\mot{e}\) an initial finite path of~\(\aut{T}\),
we also call \(\aut{T}_{|\bot(\mot{e})}\) the \emph{section of~\(\aut{T}\)
at~\(\mot{e}\)} and denote it by~\(\aut{T}_{|\mot{e}}\).

Whenever one considers a rooted tree whose edges are labeled
by elements of a finite set, the \emph{label} of a (possibly infinite)
path is the ordered sequence of labels of the edges of this path.

\section{Powers of a Mealy automaton and their connected components}\label{sec-cc}
In this section we detail the basic properties of the connected
components of the powers of a reversible Mealy automaton. The link
between these components is central in our construction.

Note that some of the notions below may be defined in a more general
framework, but as in this paper we work only with reversible
automata, we restrict our attention to this case in order to make the
explanations easier.

Let \(\aut{A}=(Q,\Sigma,\delta,\rho)\) be a reversible Mealy automaton.
We will consider the connected components
of the underlying graph of~\(\aut{A}\). By reversibility,
all the connected components are strongly connected.
One can also view the connected components of~\(\aut{A}\) as the orbits
of the action of the group~\(\pres{\dual{\aut{A}}}\) generated by
the dual automaton on~\(Q\).

Now, in order to construct the orbit tree of~\(\dual{\aut{A}}\),
we will consider the connected components of the powers of~\(\aut{A}\):
for~\(n>0\), its \emph{\(n\)-th power}~$\aut{A}^n$ is the Mealy automaton
\begin{equation*}
\aut{A}^n = \bigl( \ Q^n,\Sigma, (\delta_i\colon Q^n \rightarrow
Q^n)_{i\in \Sigma}, (\rho_{\mot{x}}\colon \Sigma \rightarrow \Sigma
)_{\mot{x}\in Q^n} \ \bigr)\enspace.
\end{equation*}
By convention, \(\aut{A}^0\) is the trivial automaton on a singleton
stateset and the alphabet \(\Sigma\).

The powers of the reversible Mealy automaton~\(\aut{A}\) are reversible and the
connected components of~\(\aut{A}^n\) coincide with the orbits of the action of~\(\pres{\dual{\aut{A}}}\) on \(Q^n\).

\begin{definition}
The \emph{connection degree} of a Mealy automaton~\(\aut{A}\), denoted by~\(\cd{A}\),
is the largest~$n$ such that~\(\aut{A}^n\) is connected.
If~\(\aut{A}\) is not connected, its connection degree is~0;
if all the powers of~\(\aut{A}\) are connected, its connection degree is infinite.
\end{definition}

Note that a connection degree of~$\aut{A}$ is infinite if and only if the group~\(\pres{\dual{\aut{A}}}\)
acts level-transitively on the tree~\(Q^*\) (\emph{i.e.} the action of~\(\pres{\dual{\aut{A}}}\)
on each level~\(Q^n\) of~$Q^*$ is transitive).

Since~\(\aut{A}\) is reversible, there is a very particular connection
between the connected components of~\(\aut{A}^n\)
and the connected components of~\(\aut{A}^{n+1}\) as highlighted
in~\cite{Kli13}. More precisely, suppose that~\(\aut{C}\) is a
connected component of~\(\aut{A}^n\) for some~\(n\) and that~\(\mot{u}\in
Q^n\) is a state of~\(\aut{C}\). Let also \(x\in Q\) be a state of
$\aut{A}$ and $\aut{D}$ be a connected component of $\aut{A}^{n+1}$
containing the state $\mot{u}x$.  Then, for any
state~\(\mot{v}\) of~\(\aut{C}\), there exists a state of~\(\aut{D}\)
prefixed with~\(\mot{v}\):
\[\exists\mot{s}\in\Sigma^*\mid \delta_{\mot{s}}(\mot{u}) =
\mot{v}\quad \text{and so}\quad  \delta_{\mot{s}}(\mot{u}x) =
\mot{v}\delta_{\rho_{\mot{u}}(\mot{s})}(x)\enspace.\]

Furthermore, if \(\mot{u}y\) is a state of~\(\aut{D}\), for some
state~\(y\in Q\) different from~\(x\), then \(\delta_{\mot{s}}(\mot{u}x)\)
and~\(\delta_{\mot{s}}(\mot{u}y)\) are two different states of~\(\aut{D}\)
prefixed with~\(\mot{v}\),
because of the reversibility of~\(\aut{A}^{n+1}\): the transition
function~\(\delta_{\rho_{\mot{u}}(\mot{s})}\) is a permutation.

Hence \(\aut{D}\) can be seen as consisting of several full
copies of~\(\aut{C}\) and~\(\#\aut{C}\) divides~\(\#\aut{D}\). They
have the same size if and only if, once fixed some state~\(\mot{u}\)
of~\(\aut{C}\), for any different states~\(x,y\in Q\), \(\mot{u}x\) and
\(\mot{u}y\) cannot both belong to \(\aut{D}\).

If from a connected component \(\aut{C}\)
of~\(\aut{A}^n\), we obtain several connected components
of~\(\aut{A}^{n+1}\), we say that~\(\aut{C}\) \emph{splits up}.

\medskip

The connected components of the powers of a Mealy automaton and the
questions of the finiteness of the generated group
or the existence of a monogenic subgroup are closely related,
as shown in the following propositions (obtained also independently in~\cite{dangeli_r:geometric}).

\begin{proposition}\label{prop-bounded-cc}
An invertible-reversible Mealy automaton generates a finite group if
and only if the connected components of its powers have bounded size.
\end{proposition}

\begin{proof}
Let \(\aut{A}\) be an invertible-reversible Mealy automaton.

Let \(\aut{C}\) be a connected component of some~\(\aut{A}^n\), and
\(\mot{u}, \mot{v}\) two states of~\(\aut{C}\):
because of the reversibility of~\(\aut{A}^n\), \(\mot{v}\) is the
image of~\(\mot{u}\) by the action of an element of~\(\pres{\dual{\aut{A}}}\).
If \(\aut{A}\) generates a finite group, so
does \(\dual{\aut{A}}\) by Proposition~\ref{prop:finite_dual}, and
hence \(\#\aut{C}\) is bounded by~\(\#\pres{\dual{\aut{A}}}\).

Conversely, if the connected components of the powers of~\(\aut{A}\) have bounded
size, as there is only a finite number of Mealy automata of bounded
size, there exist \(p<q\) such that~\(\mz(\aut{A}^p)\) and
\(\mz(\aut{A}^q)\) are isomorphic,
and therefore \(\aut{A}\) generates a finite group from
Lemma~\ref{lm:powers}.
\end{proof}

\begin{proposition}\label{prop-finite}
Let $\aut{A}=(Q,\Sigma,\delta,\rho)$ be an invertible-reversible Mealy automaton and let~\(\mot{u}\in Q^+\) be a non-empty word. The following conditions are
equivalent:
\begin{enumerate}[(i)]
\item \(\rho_{\mot{u}}\) has finite order,\label{i1}
\item the sizes of the connected components of~\((\mot{u}^n)_{n\in
  \N}\) are bounded,\label{i2}
\item there exists a word \(\mot{v}\) such that the sizes of the
  connected components of~\((\mot{vu}^n)_{n\in \N}\) are
  bounded,\label{i3}
\item for any word \(\mot{v}\), the sizes of the
  connected components of~\((\mot{vu}^n)_{n\in \N}\) are
  bounded.\label{i4}
\end{enumerate}
\end{proposition}

\begin{proof}
\eqref{i2}\(\Rightarrow\)\eqref{i3},
\eqref{i4}\(\Rightarrow\)\eqref{i2}, and
\eqref{i4}\(\Rightarrow\)\eqref{i3} are immediate.

\eqref{i1}\(\Rightarrow\)\eqref{i2}  is a direct consequence of
Proposition~\ref{prop-bounded-cc}: let \(k\) be the order of
\(\rho_{\mot{u}}\); it means that~\(\mot{u}^k\) acts as the identity,
and so do all the states of its connected component. By
Proposition~\ref{prop-bounded-cc}, the connected components of the
\((\mot{u}^{kn})_n\) have bounded size, which leads to~\eqref{i2}.

\eqref{i3}\(\Rightarrow\)\eqref{i1}: for each \(n\), denote by
\(\aut{C}_n\) the connected component of~\(\mot{vu}^n\). As the sizes
of these components are bounded, the sequence \((\aut{C}_n)_n\) admits a
subsequence which all
elements are the same, up to state numbering. Within this subsequence,
there are two elements such that two different words \(\mot{vu}^*\)
name the same state, say \(\mot{vu}^p\) and~\(\mot{vu}^q\), which means
that~\(\rho_{\mot{vu}^p}=\rho_{\mot{vu}^q}\), and \(\rho_{\mot{u}}\)
has finite order.

\eqref{i2}\(\Rightarrow\)\eqref{i4}: the size of the connected
component of \(\mot{vu}^n\) is at most \(\#\Sigma^{|\mot{v}|}\) times
the size of the connected component of~\(\mot{u}^n\).
\end{proof}

\section{The Labeled Orbit Tree}\label{sec-otree}
In this section, we build a tree capturing the links between the
connected components of consecutive powers of a Mealy
automaton. An example of the first levels of such a tree is given in
Figure~\ref{fig-otree}.

Let \(\aut{A}=(Q,\Sigma,\delta,\rho)\) be an invertible-reversible Mealy automaton.
Consider the tree with vertices
the connected components of the powers of~\(\aut{A}\), and the
incidence relation built by adding an element of~\(Q\): for any~\(n\geq 0\),
the connected component of~\(\mot{u}\in Q^n\) is linked
to the connected component(s) of~\(\mot{u}x\), for any~\(x\in
Q\). This tree is called the \emph{orbit tree} of~\(\dual{\aut{A}}\)~\cite{gawron_ns:conjugation}.
It can be seen as the
quotient of the tree \(Q^*\) under the
action of the group~\(\pres{\dual{\aut{A}}}\).

We label any
edge \(\aut{C}\to\aut{D}\) of the orbit tree by the ratio
\(\frac{\#\aut{D}}{\#\aut{C}}\), which is always an integer by the
reversibility of~\(\aut{A}\). We call this labeled tree the
\emph{labeled orbit tree} of~\(\dual{\aut{A}}\).
In~\cite{gawron_ns:conjugation}, in the definition of the labeled orbit tree,
each vertex is labeled by the size of the associated connected component,
which encodes exactly the same information as our relative labeling, as the root has size~one.
We denote by~\(\otree[A]\) the labeled orbit tree of~\(\dual{\aut{A}}\).
Note that for each vertex of \(\otree[A]\) the sum of the labels of all edges going down from this vertex
always equals to the number of states in~\(\aut{A}\).

\begin{figure}[ht]
 	\begin{center}
         \scalebox{.444}{
\begin{tikzpicture}[>=latex',line join=bevel,scale=.85]
  \definecolor{strokecolor}{rgb}{0.8,0.8,0.8};
  \definecolor{fillcolor}{rgb}{0.8,0.8,0.8};
  \node (a0_1) at (482.0bp,468.0bp) [draw=strokecolor,fill=fillcolor,ellipse] {};	
  \coordinate (a0_1left) at (472.0bp,468.0bp);	
  \definecolor{strokecolor}{rgb}{0.71,0.67,0.89};
  \definecolor{fillcolor}{rgb}{0.71,0.67,0.89};
  \node (a2_1) at (482.0bp,352.0bp) [draw=strokecolor,fill=fillcolor,ellipse] {};	
  \coordinate (a2_1left) at (472.0bp,352.0bp);	
  \node[scale=1.5] (a2_1ne) at (552.0bp,402.0bp) {vertex~\({\Lambda}\)};	
  \coordinate (a2_1bis) at (488.0bp,358.0bp);	
  \definecolor{strokecolor}{rgb}{0.41,0.01,0.27};
  \definecolor{fillcolor}{rgb}{0.41,0.01,0.27};
  \node (a3_2) at (580.0bp,294.0bp) [draw=strokecolor,fill=fillcolor,ellipse] {};	
  \coordinate (a3_2bis) at (586.0bp,300.0bp);	
  \node[scale=1.5] (a3_2ne) at (650.0bp,344.0bp) {vertex~\({\mathcal O}_2\)};		
  \coordinate (a2_1a3_2) at (528.0bp,320.0bp);		
  \node[scale=1.5] (a2_1a3_2se) at (481.0bp,273.0bp)  {edge~\(e_2\)};	
 \draw[darkgray,decorate,decoration={brace,amplitude=7pt}] (a2_1left) -- (a0_1left) node[left,midway,scale=1.5] {\(\curlywedge~\)};
 \draw[gray,->,thick] (a3_2ne) -- (a3_2bis);
 \draw[gray,->,thick] (a2_1a3_2se) -- (a2_1a3_2);
 \draw[gray,->,thick] (a2_1ne) -- (a2_1bis);
  \definecolor{strokecolor}{rgb}{0.24,0.86,0.16};
  \definecolor{fillcolor}{rgb}{0.24,0.86,0.16};
  \node (a7_24) at (666.0bp,62.0bp) [draw=strokecolor,fill=fillcolor,ellipse] {};
  \definecolor{strokecolor}{rgb}{0.24,0.86,0.16};
  \definecolor{fillcolor}{rgb}{0.24,0.86,0.16};
  \node (a7_25) at (679.0bp,62.0bp) [draw=strokecolor,fill=fillcolor,ellipse] {};
  \definecolor{strokecolor}{rgb}{0.24,0.86,0.16};
  \definecolor{fillcolor}{rgb}{0.24,0.86,0.16};
  \node (a7_26) at (705.0bp,62.0bp) [draw=strokecolor,fill=fillcolor,ellipse] {};
  \definecolor{strokecolor}{rgb}{0.64,0.82,0.67};
  \definecolor{fillcolor}{rgb}{0.64,0.82,0.67};
  \node (a7_27) at (744.0bp,62.0bp) [draw=strokecolor,fill=fillcolor,ellipse] {};
  \definecolor{strokecolor}{rgb}{0.64,0.82,0.67};
  \definecolor{fillcolor}{rgb}{0.64,0.82,0.67};
  \node (a7_20) at (549.0bp,62.0bp) [draw=strokecolor,fill=fillcolor,ellipse] {};
  \definecolor{strokecolor}{rgb}{0.41,0.01,0.27};
  \definecolor{fillcolor}{rgb}{0.41,0.01,0.27};
  \node (a7_21) at (588.0bp,62.0bp) [draw=strokecolor,fill=fillcolor,ellipse] {};
  \definecolor{strokecolor}{rgb}{0.41,0.01,0.27};
  \definecolor{fillcolor}{rgb}{0.41,0.01,0.27};
  \node (a7_22) at (606.0bp,62.0bp) [draw=strokecolor,fill=fillcolor,ellipse] {};
  \definecolor{strokecolor}{rgb}{0.41,0.01,0.27};
  \definecolor{fillcolor}{rgb}{0.41,0.01,0.27};
  \node (a7_23) at (627.0bp,62.0bp) [draw=strokecolor,fill=fillcolor,ellipse] {};
  \definecolor{strokecolor}{rgb}{0.24,0.86,0.16};
  \definecolor{fillcolor}{rgb}{0.24,0.86,0.16};
  \node (a8_78) at (1004.0bp,4.0bp) [draw=strokecolor,fill=fillcolor,ellipse] {};
  \definecolor{strokecolor}{rgb}{0.24,0.86,0.16};
  \definecolor{fillcolor}{rgb}{0.24,0.86,0.16};
  \node (a8_79) at (1017.0bp,4.0bp) [draw=strokecolor,fill=fillcolor,ellipse] {};
  \definecolor{strokecolor}{rgb}{0.89,0.3,0.95};
  \definecolor{fillcolor}{rgb}{0.89,0.3,0.95};
  \node (a7_28) at (763.0bp,62.0bp) [draw=strokecolor,fill=fillcolor,ellipse] {};
  \definecolor{strokecolor}{rgb}{0.84,0.87,0.71};
  \definecolor{fillcolor}{rgb}{0.84,0.87,0.71};
  \node (a7_29) at (796.0bp,62.0bp) [draw=strokecolor,fill=fillcolor,ellipse] {};
  \definecolor{strokecolor}{rgb}{0.41,0.01,0.27};
  \definecolor{fillcolor}{rgb}{0.41,0.01,0.27};
  \node (a8_49) at (627.0bp,4.0bp) [draw=strokecolor,fill=fillcolor,ellipse] {};
  \definecolor{strokecolor}{rgb}{0.24,0.86,0.16};
  \definecolor{fillcolor}{rgb}{0.24,0.86,0.16};
  \node (a8_48) at (614.0bp,4.0bp) [draw=strokecolor,fill=fillcolor,ellipse] {};
  \definecolor{strokecolor}{rgb}{0.41,0.01,0.27};
  \definecolor{fillcolor}{rgb}{0.41,0.01,0.27};
  \node (a8_47) at (601.0bp,4.0bp) [draw=strokecolor,fill=fillcolor,ellipse] {};
  \definecolor{strokecolor}{rgb}{0.24,0.86,0.16};
  \definecolor{fillcolor}{rgb}{0.24,0.86,0.16};
  \node (a8_46) at (588.0bp,4.0bp) [draw=strokecolor,fill=fillcolor,ellipse] {};
  \definecolor{strokecolor}{rgb}{0.41,0.01,0.27};
  \definecolor{fillcolor}{rgb}{0.41,0.01,0.27};
  \node (a8_45) at (575.0bp,4.0bp) [draw=strokecolor,fill=fillcolor,ellipse] {};
  \definecolor{strokecolor}{rgb}{0.64,0.82,0.67};
  \definecolor{fillcolor}{rgb}{0.64,0.82,0.67};
  \node (a8_44) at (562.0bp,4.0bp) [draw=strokecolor,fill=fillcolor,ellipse] {};
  \definecolor{strokecolor}{rgb}{0.89,0.3,0.95};
  \definecolor{fillcolor}{rgb}{0.89,0.3,0.95};
  \node (a8_43) at (549.0bp,4.0bp) [draw=strokecolor,fill=fillcolor,ellipse] {};
  \definecolor{strokecolor}{rgb}{0.84,0.87,0.71};
  \definecolor{fillcolor}{rgb}{0.84,0.87,0.71};
  \node (a8_42) at (536.0bp,4.0bp) [draw=strokecolor,fill=fillcolor,ellipse] {};
  \definecolor{strokecolor}{rgb}{0.64,0.82,0.67};
  \definecolor{fillcolor}{rgb}{0.64,0.82,0.67};
  \node (a8_41) at (523.0bp,4.0bp) [draw=strokecolor,fill=fillcolor,ellipse] {};
  \definecolor{strokecolor}{rgb}{0.64,0.82,0.67};
  \definecolor{fillcolor}{rgb}{0.64,0.82,0.67};
  \node (a8_40) at (510.0bp,4.0bp) [draw=strokecolor,fill=fillcolor,ellipse] {};
  \definecolor{strokecolor}{rgb}{0.24,0.86,0.16};
  \definecolor{fillcolor}{rgb}{0.24,0.86,0.16};
  \node (a6_4) at (166.0bp,120.0bp) [draw=strokecolor,fill=fillcolor,ellipse] {};
  \definecolor{strokecolor}{rgb}{0.24,0.86,0.16};
  \definecolor{fillcolor}{rgb}{0.24,0.86,0.16};
  \node (a6_5) at (243.0bp,120.0bp) [draw=strokecolor,fill=fillcolor,ellipse] {};
  \definecolor{strokecolor}{rgb}{0.41,0.01,0.27};
  \definecolor{fillcolor}{rgb}{0.41,0.01,0.27};
  \node (a6_6) at (297.0bp,120.0bp) [draw=strokecolor,fill=fillcolor,ellipse] {};
  \definecolor{strokecolor}{rgb}{0.84,0.87,0.71};
  \definecolor{fillcolor}{rgb}{0.84,0.87,0.71};
  \node (a6_7) at (361.0bp,120.0bp) [draw=strokecolor,fill=fillcolor,ellipse] {};
  \definecolor{strokecolor}{rgb}{0.71,0.67,0.89};
  \definecolor{fillcolor}{rgb}{0.71,0.67,0.89};
  \node (a6_1) at (36.0bp,120.0bp) [draw=strokecolor,fill=fillcolor,ellipse] {};
  \definecolor{strokecolor}{rgb}{0.41,0.01,0.27};
  \definecolor{fillcolor}{rgb}{0.41,0.01,0.27};
  \node (a6_2) at (75.0bp,120.0bp) [draw=strokecolor,fill=fillcolor,ellipse] {};
  \definecolor{strokecolor}{rgb}{0.41,0.01,0.27};
  \definecolor{fillcolor}{rgb}{0.41,0.01,0.27};
  \node (a6_3) at (142.0bp,120.0bp) [draw=strokecolor,fill=fillcolor,ellipse] {};
  \definecolor{strokecolor}{rgb}{0.24,0.86,0.16};
  \definecolor{fillcolor}{rgb}{0.24,0.86,0.16};
  \node (a6_8) at (406.0bp,120.0bp) [draw=strokecolor,fill=fillcolor,ellipse] {};
  \definecolor{strokecolor}{rgb}{0.24,0.86,0.16};
  \definecolor{fillcolor}{rgb}{0.24,0.86,0.16};
  \node (a6_9) at (507.0bp,120.0bp) [draw=strokecolor,fill=fillcolor,ellipse] {};
  \definecolor{strokecolor}{rgb}{0.89,0.3,0.95};
  \definecolor{fillcolor}{rgb}{0.89,0.3,0.95};
  \node (a8_58) at (744.0bp,4.0bp) [draw=strokecolor,fill=fillcolor,ellipse] {};
  \definecolor{strokecolor}{rgb}{0.89,0.3,0.95};
  \definecolor{fillcolor}{rgb}{0.89,0.3,0.95};
  \node (a8_59) at (757.0bp,4.0bp) [draw=strokecolor,fill=fillcolor,ellipse] {};
  \definecolor{strokecolor}{rgb}{0.24,0.86,0.16};
  \definecolor{fillcolor}{rgb}{0.24,0.86,0.16};
  \node (a8_50) at (640.0bp,4.0bp) [draw=strokecolor,fill=fillcolor,ellipse] {};
  \definecolor{strokecolor}{rgb}{0.84,0.87,0.71};
  \definecolor{fillcolor}{rgb}{0.84,0.87,0.71};
  \node (a8_51) at (653.0bp,4.0bp) [draw=strokecolor,fill=fillcolor,ellipse] {};
  \definecolor{strokecolor}{rgb}{0.24,0.86,0.16};
  \definecolor{fillcolor}{rgb}{0.24,0.86,0.16};
  \node (a8_52) at (666.0bp,4.0bp) [draw=strokecolor,fill=fillcolor,ellipse] {};
  \definecolor{strokecolor}{rgb}{0.84,0.87,0.71};
  \definecolor{fillcolor}{rgb}{0.84,0.87,0.71};
  \node (a8_53) at (679.0bp,4.0bp) [draw=strokecolor,fill=fillcolor,ellipse] {};
  \definecolor{strokecolor}{rgb}{0.24,0.86,0.16};
  \definecolor{fillcolor}{rgb}{0.24,0.86,0.16};
  \node (a8_54) at (692.0bp,4.0bp) [draw=strokecolor,fill=fillcolor,ellipse] {};
  \definecolor{strokecolor}{rgb}{0.84,0.87,0.71};
  \definecolor{fillcolor}{rgb}{0.84,0.87,0.71};
  \node (a8_55) at (705.0bp,4.0bp) [draw=strokecolor,fill=fillcolor,ellipse] {};
  \definecolor{strokecolor}{rgb}{0.24,0.86,0.16};
  \definecolor{fillcolor}{rgb}{0.24,0.86,0.16};
  \node (a8_56) at (718.0bp,4.0bp) [draw=strokecolor,fill=fillcolor,ellipse] {};
  \definecolor{strokecolor}{rgb}{0.64,0.82,0.67};
  \definecolor{fillcolor}{rgb}{0.64,0.82,0.67};
  \node (a8_57) at (731.0bp,4.0bp) [draw=strokecolor,fill=fillcolor,ellipse] {};
  \definecolor{strokecolor}{rgb}{0.64,0.82,0.67};
  \definecolor{fillcolor}{rgb}{0.64,0.82,0.67};
  \node (a8_72) at (926.0bp,4.0bp) [draw=strokecolor,fill=fillcolor,ellipse] {};
  \definecolor{strokecolor}{rgb}{0.24,0.86,0.16};
  \definecolor{fillcolor}{rgb}{0.24,0.86,0.16};
  \node (a8_73) at (939.0bp,4.0bp) [draw=strokecolor,fill=fillcolor,ellipse] {};
  \definecolor{strokecolor}{rgb}{0.64,0.82,0.67};
  \definecolor{fillcolor}{rgb}{0.64,0.82,0.67};
  \node (a8_70) at (900.0bp,4.0bp) [draw=strokecolor,fill=fillcolor,ellipse] {};
  \definecolor{strokecolor}{rgb}{0.71,0.67,0.89};
  \definecolor{fillcolor}{rgb}{0.71,0.67,0.89};
  \node (a3_1) at (323.0bp,294.0bp) [draw=strokecolor,fill=fillcolor,ellipse] {};
  \definecolor{strokecolor}{rgb}{0.84,0.87,0.71};
  \definecolor{fillcolor}{rgb}{0.84,0.87,0.71};
  \node (a8_71) at (913.0bp,4.0bp) [draw=strokecolor,fill=fillcolor,ellipse] {};
  \definecolor{strokecolor}{rgb}{0.24,0.86,0.16};
  \definecolor{fillcolor}{rgb}{0.24,0.86,0.16};
  \node (a8_76) at (978.0bp,4.0bp) [draw=strokecolor,fill=fillcolor,ellipse] {};
  \definecolor{strokecolor}{rgb}{0.24,0.86,0.16};
  \definecolor{fillcolor}{rgb}{0.24,0.86,0.16};
  \node (a8_77) at (991.0bp,4.0bp) [draw=strokecolor,fill=fillcolor,ellipse] {};
  \definecolor{strokecolor}{rgb}{0.24,0.86,0.16};
  \definecolor{fillcolor}{rgb}{0.24,0.86,0.16};
  \node (a8_74) at (952.0bp,4.0bp) [draw=strokecolor,fill=fillcolor,ellipse] {};
  \definecolor{strokecolor}{rgb}{0.24,0.86,0.16};
  \definecolor{fillcolor}{rgb}{0.24,0.86,0.16};
  \node (a8_75) at (965.0bp,4.0bp) [draw=strokecolor,fill=fillcolor,ellipse] {};
  \definecolor{strokecolor}{rgb}{0.84,0.87,0.71};
  \definecolor{fillcolor}{rgb}{0.84,0.87,0.71};
  \node (a8_29) at (367.0bp,4.0bp) [draw=strokecolor,fill=fillcolor,ellipse] {};
  \definecolor{strokecolor}{rgb}{0.89,0.3,0.95};
  \definecolor{fillcolor}{rgb}{0.89,0.3,0.95};
  \node (a8_28) at (354.0bp,4.0bp) [draw=strokecolor,fill=fillcolor,ellipse] {};
  \definecolor{strokecolor}{rgb}{0.24,0.86,0.16};
  \definecolor{fillcolor}{rgb}{0.24,0.86,0.16};
  \node (a8_25) at (315.0bp,4.0bp) [draw=strokecolor,fill=fillcolor,ellipse] {};
  \definecolor{strokecolor}{rgb}{0.24,0.86,0.16};
  \definecolor{fillcolor}{rgb}{0.24,0.86,0.16};
  \node (a8_24) at (302.0bp,4.0bp) [draw=strokecolor,fill=fillcolor,ellipse] {};
  \definecolor{strokecolor}{rgb}{0.64,0.82,0.67};
  \definecolor{fillcolor}{rgb}{0.64,0.82,0.67};
  \node (a8_27) at (341.0bp,4.0bp) [draw=strokecolor,fill=fillcolor,ellipse] {};
  \definecolor{strokecolor}{rgb}{0.24,0.86,0.16};
  \definecolor{fillcolor}{rgb}{0.24,0.86,0.16};
  \node (a8_26) at (328.0bp,4.0bp) [draw=strokecolor,fill=fillcolor,ellipse] {};
  \definecolor{strokecolor}{rgb}{0.41,0.01,0.27};
  \definecolor{fillcolor}{rgb}{0.41,0.01,0.27};
  \node (a8_21) at (263.0bp,4.0bp) [draw=strokecolor,fill=fillcolor,ellipse] {};
  \definecolor{strokecolor}{rgb}{0.64,0.82,0.67};
  \definecolor{fillcolor}{rgb}{0.64,0.82,0.67};
  \node (a8_20) at (250.0bp,4.0bp) [draw=strokecolor,fill=fillcolor,ellipse] {};
  \definecolor{strokecolor}{rgb}{0.41,0.01,0.27};
  \definecolor{fillcolor}{rgb}{0.41,0.01,0.27};
  \node (a8_23) at (289.0bp,4.0bp) [draw=strokecolor,fill=fillcolor,ellipse] {};
  \definecolor{strokecolor}{rgb}{0.41,0.01,0.27};
  \definecolor{fillcolor}{rgb}{0.41,0.01,0.27};
  \node (a8_22) at (276.0bp,4.0bp) [draw=strokecolor,fill=fillcolor,ellipse] {};
  \definecolor{strokecolor}{rgb}{0.24,0.86,0.16};
  \definecolor{fillcolor}{rgb}{0.24,0.86,0.16};
  \node (a8_81) at (1043.0bp,4.0bp) [draw=strokecolor,fill=fillcolor,ellipse] {};
  \definecolor{strokecolor}{rgb}{0.24,0.86,0.16};
  \definecolor{fillcolor}{rgb}{0.24,0.86,0.16};
  \node (a8_80) at (1030.0bp,4.0bp) [draw=strokecolor,fill=fillcolor,ellipse] {};
  \definecolor{strokecolor}{rgb}{0.24,0.86,0.16};
  \definecolor{fillcolor}{rgb}{0.24,0.86,0.16};
  \node (a1_1) at (482.0bp,410.0bp) [draw=strokecolor,fill=fillcolor,ellipse] {};
  \definecolor{strokecolor}{rgb}{0.24,0.86,0.16};
  \definecolor{fillcolor}{rgb}{0.24,0.86,0.16};
  \node (a7_9) at (224.0bp,62.0bp) [draw=strokecolor,fill=fillcolor,ellipse] {};
  \definecolor{strokecolor}{rgb}{0.24,0.86,0.16};
  \definecolor{fillcolor}{rgb}{0.24,0.86,0.16};
  \node (a7_8) at (185.0bp,62.0bp) [draw=strokecolor,fill=fillcolor,ellipse] {};
  \definecolor{strokecolor}{rgb}{0.24,0.86,0.16};
  \definecolor{fillcolor}{rgb}{0.24,0.86,0.16};
  \node (a7_5) at (120.0bp,62.0bp) [draw=strokecolor,fill=fillcolor,ellipse] {};
  \definecolor{strokecolor}{rgb}{0.24,0.86,0.16};
  \definecolor{fillcolor}{rgb}{0.24,0.86,0.16};
  \node (a7_4) at (81.0bp,62.0bp) [draw=strokecolor,fill=fillcolor,ellipse] {};
  \definecolor{strokecolor}{rgb}{0.84,0.87,0.71};
  \definecolor{fillcolor}{rgb}{0.84,0.87,0.71};
  \node (a7_7) at (166.0bp,62.0bp) [draw=strokecolor,fill=fillcolor,ellipse] {};
  \definecolor{strokecolor}{rgb}{0.41,0.01,0.27};
  \definecolor{fillcolor}{rgb}{0.41,0.01,0.27};
  \node (a7_6) at (142.0bp,62.0bp) [draw=strokecolor,fill=fillcolor,ellipse] {};
  \definecolor{strokecolor}{rgb}{0.71,0.67,0.89};
  \definecolor{fillcolor}{rgb}{0.71,0.67,0.89};
  \node (a7_1) at (16.0bp,62.0bp) [draw=strokecolor,fill=fillcolor,ellipse] {};
  \definecolor{strokecolor}{rgb}{0.41,0.01,0.27};
  \definecolor{fillcolor}{rgb}{0.41,0.01,0.27};
  \node (a7_3) at (68.0bp,62.0bp) [draw=strokecolor,fill=fillcolor,ellipse] {};
  \definecolor{strokecolor}{rgb}{0.41,0.01,0.27};
  \definecolor{fillcolor}{rgb}{0.41,0.01,0.27};
  \node (a7_2) at (36.0bp,62.0bp) [draw=strokecolor,fill=fillcolor,ellipse] {};
  \definecolor{strokecolor}{rgb}{0.84,0.87,0.71};
  \definecolor{fillcolor}{rgb}{0.84,0.87,0.71};
  \node (a8_38) at (484.0bp,4.0bp) [draw=strokecolor,fill=fillcolor,ellipse] {};
  \definecolor{strokecolor}{rgb}{0.84,0.87,0.71};
  \definecolor{fillcolor}{rgb}{0.84,0.87,0.71};
  \node (a8_39) at (497.0bp,4.0bp) [draw=strokecolor,fill=fillcolor,ellipse] {};
  \definecolor{strokecolor}{rgb}{0.24,0.86,0.16};
  \definecolor{fillcolor}{rgb}{0.24,0.86,0.16};
  \node (a8_36) at (458.0bp,4.0bp) [draw=strokecolor,fill=fillcolor,ellipse] {};
  \definecolor{strokecolor}{rgb}{0.24,0.86,0.16};
  \definecolor{fillcolor}{rgb}{0.24,0.86,0.16};
  \node (a8_37) at (471.0bp,4.0bp) [draw=strokecolor,fill=fillcolor,ellipse] {};
  \definecolor{strokecolor}{rgb}{0.24,0.86,0.16};
  \definecolor{fillcolor}{rgb}{0.24,0.86,0.16};
  \node (a8_34) at (432.0bp,4.0bp) [draw=strokecolor,fill=fillcolor,ellipse] {};
  \definecolor{strokecolor}{rgb}{0.24,0.86,0.16};
  \definecolor{fillcolor}{rgb}{0.24,0.86,0.16};
  \node (a8_35) at (445.0bp,4.0bp) [draw=strokecolor,fill=fillcolor,ellipse] {};
  \definecolor{strokecolor}{rgb}{0.84,0.87,0.71};
  \definecolor{fillcolor}{rgb}{0.84,0.87,0.71};
  \node (a8_32) at (406.0bp,4.0bp) [draw=strokecolor,fill=fillcolor,ellipse] {};
  \definecolor{strokecolor}{rgb}{0.84,0.87,0.71};
  \definecolor{fillcolor}{rgb}{0.84,0.87,0.71};
  \node (a8_33) at (419.0bp,4.0bp) [draw=strokecolor,fill=fillcolor,ellipse] {};
  \definecolor{strokecolor}{rgb}{0.64,0.82,0.67};
  \definecolor{fillcolor}{rgb}{0.64,0.82,0.67};
  \node (a8_30) at (380.0bp,4.0bp) [draw=strokecolor,fill=fillcolor,ellipse] {};
  \definecolor{strokecolor}{rgb}{0.84,0.87,0.71};
  \definecolor{fillcolor}{rgb}{0.84,0.87,0.71};
  \node (a8_31) at (393.0bp,4.0bp) [draw=strokecolor,fill=fillcolor,ellipse] {};
  \definecolor{strokecolor}{rgb}{0.84,0.87,0.71};
  \definecolor{fillcolor}{rgb}{0.84,0.87,0.71};
  \node (a8_62) at (796.0bp,4.0bp) [draw=strokecolor,fill=fillcolor,ellipse] {};
  \definecolor{strokecolor}{rgb}{0.41,0.01,0.27};
  \definecolor{fillcolor}{rgb}{0.41,0.01,0.27};
  \node (a5_3) at (297.0bp,178.0bp) [draw=strokecolor,fill=fillcolor,ellipse] {};
  \definecolor{strokecolor}{rgb}{0.41,0.01,0.27};
  \definecolor{fillcolor}{rgb}{0.41,0.01,0.27};
  \node (a5_2) at (153.0bp,178.0bp) [draw=strokecolor,fill=fillcolor,ellipse] {};
  \definecolor{strokecolor}{rgb}{0.71,0.67,0.89};
  \definecolor{fillcolor}{rgb}{0.71,0.67,0.89};
  \node (a5_1) at (75.0bp,178.0bp) [draw=strokecolor,fill=fillcolor,ellipse] {};
  \definecolor{strokecolor}{rgb}{0.84,0.87,0.71};
  \definecolor{fillcolor}{rgb}{0.84,0.87,0.71};
  \node (a5_7) at (782.0bp,178.0bp) [draw=strokecolor,fill=fillcolor,ellipse] {};
  \definecolor{strokecolor}{rgb}{0.41,0.01,0.27};
  \definecolor{fillcolor}{rgb}{0.41,0.01,0.27};
  \node (a5_6) at (606.0bp,178.0bp) [draw=strokecolor,fill=fillcolor,ellipse] {};
  \definecolor{strokecolor}{rgb}{0.24,0.86,0.16};
  \definecolor{fillcolor}{rgb}{0.24,0.86,0.16};
  \node (a5_5) at (542.0bp,178.0bp) [draw=strokecolor,fill=fillcolor,ellipse] {};
  \definecolor{strokecolor}{rgb}{0.24,0.86,0.16};
  \definecolor{fillcolor}{rgb}{0.24,0.86,0.16};
  \node (a5_4) at (361.0bp,178.0bp) [draw=strokecolor,fill=fillcolor,ellipse] {};
  \definecolor{strokecolor}{rgb}{0.24,0.86,0.16};
  \definecolor{fillcolor}{rgb}{0.24,0.86,0.16};
  \node (a5_8) at (893.0bp,178.0bp) [draw=strokecolor,fill=fillcolor,ellipse] {};
  \definecolor{strokecolor}{rgb}{0.84,0.87,0.71};
  \definecolor{fillcolor}{rgb}{0.84,0.87,0.71};
  \node (a7_19) at (536.0bp,62.0bp) [draw=strokecolor,fill=fillcolor,ellipse] {};
  \definecolor{strokecolor}{rgb}{0.84,0.87,0.71};
  \definecolor{fillcolor}{rgb}{0.84,0.87,0.71};
  \node (a7_18) at (507.0bp,62.0bp) [draw=strokecolor,fill=fillcolor,ellipse] {};
  \definecolor{strokecolor}{rgb}{0.84,0.87,0.71};
  \definecolor{fillcolor}{rgb}{0.84,0.87,0.71};
  \node (a7_15) at (406.0bp,62.0bp) [draw=strokecolor,fill=fillcolor,ellipse] {};
  \definecolor{strokecolor}{rgb}{0.84,0.87,0.71};
  \definecolor{fillcolor}{rgb}{0.84,0.87,0.71};
  \node (a7_14) at (367.0bp,62.0bp) [draw=strokecolor,fill=fillcolor,ellipse] {};
  \definecolor{strokecolor}{rgb}{0.24,0.86,0.16};
  \definecolor{fillcolor}{rgb}{0.24,0.86,0.16};
  \node (a7_17) at (484.0bp,62.0bp) [draw=strokecolor,fill=fillcolor,ellipse] {};
  \definecolor{strokecolor}{rgb}{0.24,0.86,0.16};
  \definecolor{fillcolor}{rgb}{0.24,0.86,0.16};
  \node (a7_16) at (439.0bp,62.0bp) [draw=strokecolor,fill=fillcolor,ellipse] {};
  \definecolor{strokecolor}{rgb}{0.41,0.01,0.27};
  \definecolor{fillcolor}{rgb}{0.41,0.01,0.27};
  \node (a7_11) at (282.0bp,62.0bp) [draw=strokecolor,fill=fillcolor,ellipse] {};
  \definecolor{strokecolor}{rgb}{0.84,0.87,0.71};
  \definecolor{fillcolor}{rgb}{0.84,0.87,0.71};
  \node (a7_10) at (243.0bp,62.0bp) [draw=strokecolor,fill=fillcolor,ellipse] {};
  \definecolor{strokecolor}{rgb}{0.64,0.82,0.67};
  \definecolor{fillcolor}{rgb}{0.64,0.82,0.67};
  \node (a7_13) at (354.0bp,62.0bp) [draw=strokecolor,fill=fillcolor,ellipse] {};
  \definecolor{strokecolor}{rgb}{0.24,0.86,0.16};
  \definecolor{fillcolor}{rgb}{0.24,0.86,0.16};
  \node (a7_12) at (309.0bp,62.0bp) [draw=strokecolor,fill=fillcolor,ellipse] {};
  \definecolor{strokecolor}{rgb}{0.24,0.86,0.16};
  \definecolor{fillcolor}{rgb}{0.24,0.86,0.16};
  \node (a8_8) at (94.0bp,4.0bp) [draw=strokecolor,fill=fillcolor,ellipse] {};
  \definecolor{strokecolor}{rgb}{0.24,0.86,0.16};
  \definecolor{fillcolor}{rgb}{0.24,0.86,0.16};
  \node (a8_9) at (107.0bp,4.0bp) [draw=strokecolor,fill=fillcolor,ellipse] {};
  \definecolor{strokecolor}{rgb}{0.41,0.01,0.27};
  \definecolor{fillcolor}{rgb}{0.41,0.01,0.27};
  \node (a8_6) at (68.0bp,4.0bp) [draw=strokecolor,fill=fillcolor,ellipse] {};
  \definecolor{strokecolor}{rgb}{0.84,0.87,0.71};
  \definecolor{fillcolor}{rgb}{0.84,0.87,0.71};
  \node (a8_7) at (81.0bp,4.0bp) [draw=strokecolor,fill=fillcolor,ellipse] {};
  \definecolor{strokecolor}{rgb}{0.24,0.86,0.16};
  \definecolor{fillcolor}{rgb}{0.24,0.86,0.16};
  \node (a8_4) at (42.0bp,4.0bp) [draw=strokecolor,fill=fillcolor,ellipse] {};
  \definecolor{strokecolor}{rgb}{0.24,0.86,0.16};
  \definecolor{fillcolor}{rgb}{0.24,0.86,0.16};
  \node (a8_5) at (55.0bp,4.0bp) [draw=strokecolor,fill=fillcolor,ellipse] {};
  \definecolor{strokecolor}{rgb}{0.41,0.01,0.27};
  \definecolor{fillcolor}{rgb}{0.41,0.01,0.27};
  \node (a8_2) at (16.0bp,4.0bp) [draw=strokecolor,fill=fillcolor,ellipse] {};
  \definecolor{strokecolor}{rgb}{0.41,0.01,0.27};
  \definecolor{fillcolor}{rgb}{0.41,0.01,0.27};
  \node (a8_3) at (29.0bp,4.0bp) [draw=strokecolor,fill=fillcolor,ellipse] {};
  \definecolor{strokecolor}{rgb}{0.71,0.67,0.89};
  \definecolor{fillcolor}{rgb}{0.71,0.67,0.89};
  \node (a8_1) at (3.0bp,4.0bp) [draw=strokecolor,fill=fillcolor,ellipse] {};
  \definecolor{strokecolor}{rgb}{0.84,0.87,0.71};
  \definecolor{fillcolor}{rgb}{0.84,0.87,0.71};
  \node (a8_14) at (172.0bp,4.0bp) [draw=strokecolor,fill=fillcolor,ellipse] {};
  \definecolor{strokecolor}{rgb}{0.84,0.87,0.71};
  \definecolor{fillcolor}{rgb}{0.84,0.87,0.71};
  \node (a8_15) at (185.0bp,4.0bp) [draw=strokecolor,fill=fillcolor,ellipse] {};
  \definecolor{strokecolor}{rgb}{0.24,0.86,0.16};
  \definecolor{fillcolor}{rgb}{0.24,0.86,0.16};
  \node (a8_16) at (198.0bp,4.0bp) [draw=strokecolor,fill=fillcolor,ellipse] {};
  \definecolor{strokecolor}{rgb}{0.24,0.86,0.16};
  \definecolor{fillcolor}{rgb}{0.24,0.86,0.16};
  \node (a8_17) at (211.0bp,4.0bp) [draw=strokecolor,fill=fillcolor,ellipse] {};
  \definecolor{strokecolor}{rgb}{0.84,0.87,0.71};
  \definecolor{fillcolor}{rgb}{0.84,0.87,0.71};
  \node (a8_10) at (120.0bp,4.0bp) [draw=strokecolor,fill=fillcolor,ellipse] {};
  \definecolor{strokecolor}{rgb}{0.41,0.01,0.27};
  \definecolor{fillcolor}{rgb}{0.41,0.01,0.27};
  \node (a8_11) at (133.0bp,4.0bp) [draw=strokecolor,fill=fillcolor,ellipse] {};
  \definecolor{strokecolor}{rgb}{0.24,0.86,0.16};
  \definecolor{fillcolor}{rgb}{0.24,0.86,0.16};
  \node (a8_12) at (146.0bp,4.0bp) [draw=strokecolor,fill=fillcolor,ellipse] {};
  \definecolor{strokecolor}{rgb}{0.64,0.82,0.67};
  \definecolor{fillcolor}{rgb}{0.64,0.82,0.67};
  \node (a8_13) at (159.0bp,4.0bp) [draw=strokecolor,fill=fillcolor,ellipse] {};
  \definecolor{strokecolor}{rgb}{0.84,0.87,0.71};
  \definecolor{fillcolor}{rgb}{0.84,0.87,0.71};
  \node (a8_18) at (224.0bp,4.0bp) [draw=strokecolor,fill=fillcolor,ellipse] {};
  \definecolor{strokecolor}{rgb}{0.84,0.87,0.71};
  \definecolor{fillcolor}{rgb}{0.84,0.87,0.71};
  \node (a8_19) at (237.0bp,4.0bp) [draw=strokecolor,fill=fillcolor,ellipse] {};
  \definecolor{strokecolor}{rgb}{0.24,0.86,0.16};
  \definecolor{fillcolor}{rgb}{0.24,0.86,0.16};
  \node (a6_16) at (972.0bp,120.0bp) [draw=strokecolor,fill=fillcolor,ellipse] {};
  \definecolor{strokecolor}{rgb}{0.84,0.87,0.71};
  \definecolor{fillcolor}{rgb}{0.84,0.87,0.71};
  \node (a6_14) at (796.0bp,120.0bp) [draw=strokecolor,fill=fillcolor,ellipse] {};
  \definecolor{strokecolor}{rgb}{0.84,0.87,0.71};
  \definecolor{fillcolor}{rgb}{0.84,0.87,0.71};
  \node (a6_15) at (893.0bp,120.0bp) [draw=strokecolor,fill=fillcolor,ellipse] {};
  \definecolor{strokecolor}{rgb}{0.24,0.86,0.16};
  \definecolor{fillcolor}{rgb}{0.24,0.86,0.16};
  \node (a6_12) at (672.0bp,120.0bp) [draw=strokecolor,fill=fillcolor,ellipse] {};
  \definecolor{strokecolor}{rgb}{0.64,0.82,0.67};
  \definecolor{fillcolor}{rgb}{0.64,0.82,0.67};
  \node (a6_13) at (763.0bp,120.0bp) [draw=strokecolor,fill=fillcolor,ellipse] {};
  \definecolor{strokecolor}{rgb}{0.84,0.87,0.71};
  \definecolor{fillcolor}{rgb}{0.84,0.87,0.71};
  \node (a6_10) at (542.0bp,120.0bp) [draw=strokecolor,fill=fillcolor,ellipse] {};
  \definecolor{strokecolor}{rgb}{0.41,0.01,0.27};
  \definecolor{fillcolor}{rgb}{0.41,0.01,0.27};
  \node (a6_11) at (606.0bp,120.0bp) [draw=strokecolor,fill=fillcolor,ellipse] {};
  \definecolor{strokecolor}{rgb}{0.41,0.01,0.27};
  \definecolor{fillcolor}{rgb}{0.41,0.01,0.27};
  \node (a4_2) at (323.0bp,236.0bp) [draw=strokecolor,fill=fillcolor,ellipse] {};
  \definecolor{strokecolor}{rgb}{0.41,0.01,0.27};
  \definecolor{fillcolor}{rgb}{0.41,0.01,0.27};
  \node (a4_3) at (580.0bp,236.0bp) [draw=strokecolor,fill=fillcolor,ellipse] {};
  \definecolor{strokecolor}{rgb}{0.84,0.87,0.71};
  \definecolor{fillcolor}{rgb}{0.84,0.87,0.71};
  \node (a8_61) at (783.0bp,4.0bp) [draw=strokecolor,fill=fillcolor,ellipse] {};
  \definecolor{strokecolor}{rgb}{0.75,0.89,0.66};
  \definecolor{fillcolor}{rgb}{0.75,0.89,0.66};
  \node (a8_60) at (770.0bp,4.0bp) [draw=strokecolor,fill=fillcolor,ellipse] {};
  \definecolor{strokecolor}{rgb}{0.84,0.87,0.71};
  \definecolor{fillcolor}{rgb}{0.84,0.87,0.71};
  \node (a8_63) at (809.0bp,4.0bp) [draw=strokecolor,fill=fillcolor,ellipse] {};
  \definecolor{strokecolor}{rgb}{0.71,0.67,0.89};
  \definecolor{fillcolor}{rgb}{0.71,0.67,0.89};
  \node (a4_1) at (153.0bp,236.0bp) [draw=strokecolor,fill=fillcolor,ellipse] {};
  \definecolor{strokecolor}{rgb}{0.64,0.82,0.67};
  \definecolor{fillcolor}{rgb}{0.64,0.82,0.67};
  \node (a8_65) at (835.0bp,4.0bp) [draw=strokecolor,fill=fillcolor,ellipse] {};
  \definecolor{strokecolor}{rgb}{0.64,0.82,0.67};
  \definecolor{fillcolor}{rgb}{0.64,0.82,0.67};
  \node (a8_64) at (822.0bp,4.0bp) [draw=strokecolor,fill=fillcolor,ellipse] {};
  \definecolor{strokecolor}{rgb}{0.84,0.87,0.71};
  \definecolor{fillcolor}{rgb}{0.84,0.87,0.71};
  \node (a8_67) at (861.0bp,4.0bp) [draw=strokecolor,fill=fillcolor,ellipse] {};
  \definecolor{strokecolor}{rgb}{0.64,0.82,0.67};
  \definecolor{fillcolor}{rgb}{0.64,0.82,0.67};
  \node (a8_66) at (848.0bp,4.0bp) [draw=strokecolor,fill=fillcolor,ellipse] {};
  \definecolor{strokecolor}{rgb}{0.84,0.87,0.71};
  \definecolor{fillcolor}{rgb}{0.84,0.87,0.71};
  \node (a8_69) at (887.0bp,4.0bp) [draw=strokecolor,fill=fillcolor,ellipse] {};
  \definecolor{strokecolor}{rgb}{0.64,0.82,0.67};
  \definecolor{fillcolor}{rgb}{0.64,0.82,0.67};
  \node (a8_68) at (874.0bp,4.0bp) [draw=strokecolor,fill=fillcolor,ellipse] {};
  \definecolor{strokecolor}{rgb}{0.24,0.86,0.16};
  \definecolor{fillcolor}{rgb}{0.24,0.86,0.16};
  \node (a4_4) at (782.0bp,236.0bp) [draw=strokecolor,fill=fillcolor,ellipse] {};
  \definecolor{strokecolor}{rgb}{0.84,0.87,0.71};
  \definecolor{fillcolor}{rgb}{0.84,0.87,0.71};
  \node (a7_33) at (913.0bp,62.0bp) [draw=strokecolor,fill=fillcolor,ellipse] {};
  \definecolor{strokecolor}{rgb}{0.84,0.87,0.71};
  \definecolor{fillcolor}{rgb}{0.84,0.87,0.71};
  \node (a7_32) at (893.0bp,62.0bp) [draw=strokecolor,fill=fillcolor,ellipse] {};
  \definecolor{strokecolor}{rgb}{0.84,0.87,0.71};
  \definecolor{fillcolor}{rgb}{0.84,0.87,0.71};
  \node (a7_31) at (874.0bp,62.0bp) [draw=strokecolor,fill=fillcolor,ellipse] {};
  \definecolor{strokecolor}{rgb}{0.64,0.82,0.67};
  \definecolor{fillcolor}{rgb}{0.64,0.82,0.67};
  \node (a7_30) at (829.0bp,62.0bp) [draw=strokecolor,fill=fillcolor,ellipse] {};
  \definecolor{strokecolor}{rgb}{0.24,0.86,0.16};
  \definecolor{fillcolor}{rgb}{0.24,0.86,0.16};
  \node (a7_36) at (1023.0bp,62.0bp) [draw=strokecolor,fill=fillcolor,ellipse] {};
  \definecolor{strokecolor}{rgb}{0.24,0.86,0.16};
  \definecolor{fillcolor}{rgb}{0.24,0.86,0.16};
  \node (a7_35) at (984.0bp,62.0bp) [draw=strokecolor,fill=fillcolor,ellipse] {};
  \definecolor{strokecolor}{rgb}{0.24,0.86,0.16};
  \definecolor{fillcolor}{rgb}{0.24,0.86,0.16};
  \node (a7_34) at (958.0bp,62.0bp) [draw=strokecolor,fill=fillcolor,ellipse] {};
  \draw [] (a7_11) ..controls (280.69bp,48.773bp) and (277.27bp,16.821bp)  .. (a8_22);
  \definecolor{strokecol}{rgb}{0.0,0.0,0.0};
  \pgfsetstrokecolor{strokecol}
  \draw (283.5bp,33.0bp) node {1};
  \draw [] (a7_16) ..controls (437.52bp,54.312bp) and (435.91bp,46.618bp)  .. (435.0bp,40.0bp) .. controls (433.34bp,27.917bp) and (432.44bp,13.351bp)  .. (a8_34);
  \draw (438.5bp,33.0bp) node {1};
  \draw [] (a6_4) ..controls (169.91bp,107.47bp) and (180.92bp,75.02bp)  .. (a7_8);
  \draw (181.5bp,91.0bp) node {1};
  \draw [] (a7_36) ..controls (1024.5bp,48.773bp) and (1028.5bp,16.821bp)  .. (a8_80);
  \draw (1030.5bp,33.0bp) node {1};
  \draw [] (a7_35) ..controls (982.69bp,48.773bp) and (979.27bp,16.821bp)  .. (a8_76);
  \draw (984.5bp,33.0bp) node {1};
  \draw [] (a7_17) ..controls (481.58bp,54.707bp) and (478.72bp,46.859bp)  .. (477.0bp,40.0bp) .. controls (474.0bp,28.04bp) and (472.03bp,13.405bp)  .. (a8_37);
  \draw (480.5bp,33.0bp) node {1};
  \draw [] (a7_34) ..controls (954.16bp,54.845bp) and (949.62bp,47.104bp)  .. (947.0bp,40.0bp) .. controls (942.68bp,28.305bp) and (940.23bp,13.521bp)  .. (a8_73);
  \draw (950.5bp,33.0bp) node {1};
  \draw [] (a7_3) ..controls (65.584bp,54.707bp) and (62.72bp,46.859bp)  .. (61.0bp,40.0bp) .. controls (58.0bp,28.04bp) and (56.025bp,13.405bp)  .. (a8_5);
  \draw (64.5bp,33.0bp) node {2};
  \draw [] (a7_25) ..controls (679.0bp,48.773bp) and (679.0bp,16.821bp)  .. (a8_53);
  \draw (682.5bp,33.0bp) node {2};
  \draw [] (a7_20) ..controls (549.0bp,48.773bp) and (549.0bp,16.821bp)  .. (a8_43);
  \draw (552.5bp,33.0bp) node {2};
  \draw [] (a7_10) ..controls (244.53bp,48.773bp) and (248.52bp,16.821bp)  .. (a8_20);
  \draw (251.5bp,33.0bp) node {2};
  \draw [] (a7_13) ..controls (354.0bp,48.773bp) and (354.0bp,16.821bp)  .. (a8_28);
  \draw (357.5bp,33.0bp) node {2};
  \draw [] (a6_5) ..controls (243.0bp,106.77bp) and (243.0bp,74.821bp)  .. (a7_10);
  \draw (246.5bp,91.0bp) node {2};
  \draw [] (a7_11) ..controls (278.59bp,54.777bp) and (274.53bp,46.985bp)  .. (272.0bp,40.0bp) .. controls (267.72bp,28.204bp) and (264.63bp,13.477bp)  .. (a8_21);
  \draw (275.5bp,33.0bp) node {1};
  \draw [] (a4_4) ..controls (800.1bp,225.87bp) and (875.25bp,187.96bp)  .. (a5_8);
  \draw (853.5bp,207.0bp) node {1};
  \draw [] (a7_5) ..controls (117.58bp,54.707bp) and (114.72bp,46.859bp)  .. (113.0bp,40.0bp) .. controls (110.0bp,28.04bp) and (108.03bp,13.405bp)  .. (a8_9);
  \draw (116.5bp,33.0bp) node {1};
  \draw [] (a6_16) ..controls (974.62bp,106.77bp) and (981.47bp,74.821bp)  .. (a7_35);
  \draw (982.5bp,91.0bp) node {1};
  \draw [] (a7_4) ..controls (83.416bp,54.707bp) and (86.28bp,46.859bp)  .. (88.0bp,40.0bp) .. controls (91.0bp,28.04bp) and (92.975bp,13.405bp)  .. (a8_8);
  \draw (94.5bp,33.0bp) node {1};
  \draw [] (a7_18) ..controls (504.86bp,54.698bp) and (502.36bp,46.843bp)  .. (501.0bp,40.0bp) .. controls (498.62bp,28.002bp) and (497.52bp,13.389bp)  .. (a8_39);
  \draw (504.5bp,33.0bp) node {1};
  \draw [] (a7_35) ..controls (985.53bp,48.773bp) and (989.52bp,16.821bp)  .. (a8_77);
  \draw (991.5bp,33.0bp) node {1};
  \draw [] (a7_8) ..controls (185.0bp,48.773bp) and (185.0bp,16.821bp)  .. (a8_15);
  \draw (188.5bp,33.0bp) node {2};
  \draw [] (a6_10) ..controls (540.69bp,106.77bp) and (537.27bp,74.821bp)  .. (a7_19);
  \draw (542.5bp,91.0bp) node {1};
  \draw [] (a6_14) ..controls (796.0bp,106.77bp) and (796.0bp,74.821bp)  .. (a7_29);
  \draw (799.5bp,91.0bp) node {1};
  \draw [] (a4_4) ..controls (782.0bp,222.77bp) and (782.0bp,190.82bp)  .. (a5_7);
  \draw (785.5bp,207.0bp) node {2};
  \draw [] (a7_32) ..controls (891.69bp,48.773bp) and (888.27bp,16.821bp)  .. (a8_69);
  \draw (893.5bp,33.0bp) node {1};
  \draw [] (a7_31) ..controls (871.58bp,54.707bp) and (868.72bp,46.859bp)  .. (867.0bp,40.0bp) .. controls (864.0bp,28.04bp) and (862.03bp,13.405bp)  .. (a8_67);
  \draw (870.5bp,33.0bp) node {1};
  \draw [] (a6_9) ..controls (502.26bp,107.47bp) and (488.94bp,75.02bp)  .. (a7_17);
  \draw (501.5bp,91.0bp) node {1};
  \draw [] (a7_19) ..controls (533.58bp,54.707bp) and (530.72bp,46.859bp)  .. (529.0bp,40.0bp) .. controls (526.0bp,28.04bp) and (524.03bp,13.405bp)  .. (a8_41);
  \draw (532.5bp,33.0bp) node {2};
  \draw [] (a7_16) ..controls (440.31bp,48.773bp) and (443.73bp,16.821bp)  .. (a8_35);
  \draw (445.5bp,33.0bp) node {1};
  \draw [] (a7_15) ..controls (408.42bp,54.707bp) and (411.28bp,46.859bp)  .. (413.0bp,40.0bp) .. controls (416.0bp,28.04bp) and (417.97bp,13.405bp)  .. (a8_33);
  \draw (418.5bp,33.0bp) node {1};
  \draw [] (a3_2) ..controls (580.0bp,280.77bp) and (580.0bp,248.82bp)  .. (a4_3);
  \draw (583.5bp,265.0bp) node {1};
  \draw [] (a7_6) ..controls (142.87bp,48.773bp) and (145.16bp,16.821bp)  .. (a8_12);
  \draw (148.5bp,33.0bp) node {2};
  \draw [] (a6_7) ..controls (362.31bp,106.77bp) and (365.73bp,74.821bp)  .. (a7_14);
  \draw (368.5bp,91.0bp) node {1};
  \draw [] (a7_30) ..controls (827.52bp,54.312bp) and (825.91bp,46.618bp)  .. (825.0bp,40.0bp) .. controls (823.34bp,27.917bp) and (822.44bp,13.351bp)  .. (a8_64);
  \draw (828.5bp,33.0bp) node {1};
  \draw [] (a6_12) ..controls (678.85bp,107.38bp) and (698.27bp,74.423bp)  .. (a7_26);
  \draw (695.5bp,91.0bp) node {1};
  \draw [] (a6_13) ..controls (759.09bp,107.47bp) and (748.08bp,75.02bp)  .. (a7_27);
  \draw (758.5bp,91.0bp) node {1};
  \draw [] (a6_15) ..controls (897.12bp,107.47bp) and (908.71bp,75.02bp)  .. (a7_33);
  \draw (908.5bp,91.0bp) node {1};
  \draw [] (a5_5) ..controls (534.74bp,165.38bp) and (514.14bp,132.42bp)  .. (a6_9);
  \draw (531.5bp,149.0bp) node {1};
  \draw [] (a5_6) ..controls (618.22bp,166.63bp) and (659.55bp,131.56bp)  .. (a6_12);
  \draw (649.5bp,149.0bp) node {2};
  \draw [] (a5_2) ..controls (155.84bp,164.77bp) and (163.26bp,132.82bp)  .. (a6_4);
  \draw (165.5bp,149.0bp) node {2};
  \draw [] (a5_8) ..controls (893.0bp,164.77bp) and (893.0bp,132.82bp)  .. (a6_15);
  \draw (896.5bp,149.0bp) node {2};
  \draw [] (a6_5) ..controls (239.09bp,107.47bp) and (228.08bp,75.02bp)  .. (a7_9);
  \draw (239.5bp,91.0bp) node {1};
  \draw [] (a7_13) ..controls (351.58bp,54.707bp) and (348.72bp,46.859bp)  .. (347.0bp,40.0bp) .. controls (344.0bp,28.04bp) and (342.03bp,13.405bp)  .. (a8_27);
  \draw (350.5bp,33.0bp) node {1};
  \draw [] (a7_12) ..controls (310.31bp,48.773bp) and (313.73bp,16.821bp)  .. (a8_25);
  \draw (316.5bp,33.0bp) node {1};
  \draw [] (a7_10) ..controls (241.69bp,48.773bp) and (238.27bp,16.821bp)  .. (a8_19);
  \draw (244.5bp,33.0bp) node {1};
  \draw [] (a7_6) ..controls (140.25bp,54.669bp) and (138.19bp,46.792bp)  .. (137.0bp,40.0bp) .. controls (134.89bp,27.952bp) and (133.63bp,13.367bp)  .. (a8_11);
  \draw (140.5bp,33.0bp) node {1};
  \draw [] (a7_7) ..controls (164.52bp,54.312bp) and (162.91bp,46.618bp)  .. (162.0bp,40.0bp) .. controls (160.34bp,27.917bp) and (159.44bp,13.351bp)  .. (a8_13);
  \draw (165.5bp,33.0bp) node {2};
  \draw [] (a7_33) ..controls (915.42bp,54.707bp) and (918.28bp,46.859bp)  .. (920.0bp,40.0bp) .. controls (923.0bp,28.04bp) and (924.97bp,13.405bp)  .. (a8_72);
  \draw (925.5bp,33.0bp) node {2};
  \draw [] (a7_9) ..controls (221.58bp,54.707bp) and (218.72bp,46.859bp)  .. (217.0bp,40.0bp) .. controls (214.0bp,28.04bp) and (212.03bp,13.405bp)  .. (a8_17);
  \draw (220.5bp,33.0bp) node {1};
  \draw [] (a7_33) ..controls (913.0bp,48.773bp) and (913.0bp,16.821bp)  .. (a8_71);
  \draw (916.5bp,33.0bp) node {1};
  \draw [] (a6_11) ..controls (610.33bp,107.47bp) and (622.49bp,75.02bp)  .. (a7_23);
  \draw (621.5bp,91.0bp) node {1};
  \draw [] (a6_1) ..controls (36.0bp,106.77bp) and (36.0bp,74.821bp)  .. (a7_2);
  \draw (39.5bp,91.0bp) node {2};
  \draw [] (a7_36) ..controls (1026.8bp,54.825bp) and (1031.3bp,47.068bp)  .. (1034.0bp,40.0bp) .. controls (1038.5bp,28.288bp) and (1041.5bp,13.514bp)  .. (a8_81);
  \draw (1041.5bp,33.0bp) node {1};
  \draw [] (a3_2) ..controls (605.49bp,285.93bp) and (755.64bp,244.31bp)  .. (a4_4);
  \draw (708.5bp,265.0bp) node {2};
  \draw [] (a7_22) ..controls (604.91bp,48.773bp) and (602.06bp,16.821bp)  .. (a8_47);
  \draw (607.5bp,33.0bp) node {1};
  \draw [] (a7_24) ..controls (663.58bp,54.707bp) and (660.72bp,46.859bp)  .. (659.0bp,40.0bp) .. controls (656.0bp,28.04bp) and (654.03bp,13.405bp)  .. (a8_51);
  \draw (662.5bp,33.0bp) node {2};
  \draw [] (a7_23) ..controls (629.42bp,54.707bp) and (632.28bp,46.859bp)  .. (634.0bp,40.0bp) .. controls (637.0bp,28.04bp) and (638.97bp,13.405bp)  .. (a8_50);
  \draw (639.5bp,33.0bp) node {2};
  \draw [] (a7_3) ..controls (68.0bp,48.773bp) and (68.0bp,16.821bp)  .. (a8_6);
  \draw (71.5bp,33.0bp) node {1};
  \draw [] (a7_29) ..controls (793.58bp,54.707bp) and (790.72bp,46.859bp)  .. (789.0bp,40.0bp) .. controls (786.0bp,28.04bp) and (784.03bp,13.405bp)  .. (a8_61);
  \draw (792.5bp,33.0bp) node {1};
  \draw [] (a7_1) ..controls (13.584bp,54.707bp) and (10.72bp,46.859bp)  .. (9.0bp,40.0bp) .. controls (6.0001bp,28.04bp) and (4.0254bp,13.405bp)  .. (a8_1);
  \draw (12.5bp,33.0bp) node {1};
  \draw [] (a7_2) ..controls (34.519bp,54.312bp) and (32.911bp,46.618bp)  .. (32.0bp,40.0bp) .. controls (30.336bp,27.917bp) and (29.441bp,13.351bp)  .. (a8_3);
  \draw (35.5bp,33.0bp) node {1};
  \draw [] (a7_28) ..controls (764.53bp,48.773bp) and (768.52bp,16.821bp)  .. (a8_60);
  \draw (770.5bp,33.0bp) node {2};
  \draw [] (a7_16) ..controls (442.41bp,54.777bp) and (446.47bp,46.985bp)  .. (449.0bp,40.0bp) .. controls (453.28bp,28.204bp) and (456.37bp,13.477bp)  .. (a8_36);
  \draw (456.5bp,33.0bp) node {1};
  \draw [] (a6_13) ..controls (763.0bp,106.77bp) and (763.0bp,74.821bp)  .. (a7_28);
  \draw (766.5bp,91.0bp) node {2};
  \draw [] (a6_16) ..controls (968.94bp,106.77bp) and (960.96bp,74.821bp)  .. (a7_34);
  \draw (969.5bp,91.0bp) node {1};
  \draw [] (a6_6) ..controls (299.62bp,106.77bp) and (306.47bp,74.821bp)  .. (a7_12);
  \draw (308.5bp,91.0bp) node {2};
  \draw [] (a7_22) ..controls (607.75bp,48.773bp) and (612.31bp,16.821bp)  .. (a8_48);
  \draw (613.5bp,33.0bp) node {2};
  \draw [] (a7_20) ..controls (551.42bp,54.707bp) and (554.28bp,46.859bp)  .. (556.0bp,40.0bp) .. controls (559.0bp,28.04bp) and (560.97bp,13.405bp)  .. (a8_44);
  \draw (561.5bp,33.0bp) node {1};
  \draw [] (a6_2) ..controls (76.31bp,106.77bp) and (79.733bp,74.821bp)  .. (a7_4);
  \draw (82.5bp,91.0bp) node {2};
  \draw [] (a7_34) ..controls (959.53bp,48.773bp) and (963.52bp,16.821bp)  .. (a8_75);
  \draw (965.5bp,33.0bp) node {1};
  \draw [] (a7_29) ..controls (796.0bp,48.773bp) and (796.0bp,16.821bp)  .. (a8_62);
  \draw (799.5bp,33.0bp) node {1};
  \draw [] (a5_7) ..controls (778.09bp,165.47bp) and (767.08bp,133.02bp)  .. (a6_13);
  \draw (777.5bp,149.0bp) node {2};
  \draw [] (a3_1) ..controls (299.33bp,285.2bp) and (175.88bp,244.54bp)  .. (a4_1);
  \draw (262.5bp,265.0bp) node {1};
  \draw [] (a5_2) ..controls (150.6bp,164.77bp) and (144.32bp,132.82bp)  .. (a6_3);
  \draw (152.5bp,149.0bp) node {1};
  \draw [] (a7_5) ..controls (120.0bp,48.773bp) and (120.0bp,16.821bp)  .. (a8_10);
  \draw (123.5bp,33.0bp) node {2};
  \draw [] (a2_1) ..controls (459.94bp,343.23bp) and (345.38bp,302.88bp)  .. (a3_1);
  \draw (424.5bp,323.0bp) node {1};
  \draw [] (a7_32) ..controls (894.53bp,48.773bp) and (898.52bp,16.821bp)  .. (a8_70);
  \draw (900.5bp,33.0bp) node {2};
  \draw [] (a6_8) ..controls (406.0bp,106.77bp) and (406.0bp,74.821bp)  .. (a7_15);
  \draw (409.5bp,91.0bp) node {2};
  \draw [] (a7_21) ..controls (588.0bp,48.773bp) and (588.0bp,16.821bp)  .. (a8_46);
  \draw (591.5bp,33.0bp) node {2};
  \draw [] (a7_28) ..controls (761.69bp,48.773bp) and (758.27bp,16.821bp)  .. (a8_59);
  \draw (763.5bp,33.0bp) node {1};
  \draw [] (a6_8) ..controls (412.85bp,107.38bp) and (432.27bp,74.423bp)  .. (a7_16);
  \draw (429.5bp,91.0bp) node {1};
  \draw [] (a5_6) ..controls (606.0bp,164.77bp) and (606.0bp,132.82bp)  .. (a6_11);
  \draw (609.5bp,149.0bp) node {1};
  \draw [] (a7_2) ..controls (37.31bp,48.773bp) and (40.733bp,16.821bp)  .. (a8_4);
  \draw (43.5bp,33.0bp) node {2};
  \draw [] (a5_3) ..controls (285.72bp,165.3bp) and (253.43bp,131.82bp)  .. (a6_5);
  \draw (280.5bp,149.0bp) node {2};
  \draw [] (a5_4) ..controls (361.0bp,164.77bp) and (361.0bp,132.82bp)  .. (a6_7);
  \draw (364.5bp,149.0bp) node {2};
  \draw [] (a3_1) ..controls (323.0bp,280.77bp) and (323.0bp,248.82bp)  .. (a4_2);
  \draw (326.5bp,265.0bp) node {2};
  \draw [] (a7_30) ..controls (830.31bp,48.773bp) and (833.73bp,16.821bp)  .. (a8_65);
  \draw (835.5bp,33.0bp) node {1};
  \draw [] (a5_4) ..controls (370.4bp,165.3bp) and (397.3bp,131.82bp)  .. (a6_8);
  \draw (392.5bp,149.0bp) node {1};
  \draw [] (a7_21) ..controls (585.58bp,54.707bp) and (582.72bp,46.859bp)  .. (581.0bp,40.0bp) .. controls (578.0bp,28.04bp) and (576.03bp,13.405bp)  .. (a8_45);
  \draw (584.5bp,33.0bp) node {1};
  \draw [] (a2_1) ..controls (497.86bp,341.94bp) and (563.18bp,304.61bp)  .. (a3_2);
  \draw (545.5bp,323.0bp) node {2};
  \draw [] (a6_14) ..controls (802.85bp,107.38bp) and (822.27bp,74.423bp)  .. (a7_30);
  \draw (819.5bp,91.0bp) node {2};
  \draw [] (a5_5) ..controls (542.0bp,164.77bp) and (542.0bp,132.82bp)  .. (a6_10);
  \draw (545.5bp,149.0bp) node {2};
  \draw [] (a4_3) ..controls (572.12bp,223.38bp) and (549.75bp,190.42bp)  .. (a5_5);
  \draw (568.5bp,207.0bp) node {2};
  \draw [] (a7_12) ..controls (312.84bp,54.845bp) and (317.38bp,47.104bp)  .. (320.0bp,40.0bp) .. controls (324.32bp,28.305bp) and (326.77bp,13.521bp)  .. (a8_26);
  \draw (327.5bp,33.0bp) node {1};
  \draw [] (a7_17) ..controls (484.0bp,48.773bp) and (484.0bp,16.821bp)  .. (a8_38);
  \draw (487.5bp,33.0bp) node {2};
  \draw [] (a6_1) ..controls (31.881bp,107.47bp) and (20.293bp,75.02bp)  .. (a7_1);
  \draw (32.5bp,91.0bp) node {1};
  \draw [] (a6_12) ..controls (670.69bp,106.77bp) and (667.27bp,74.821bp)  .. (a7_24);
  \draw (672.5bp,91.0bp) node {1};
  \draw [] (a7_27) ..controls (741.58bp,54.707bp) and (738.72bp,46.859bp)  .. (737.0bp,40.0bp) .. controls (734.0bp,28.04bp) and (732.03bp,13.405bp)  .. (a8_57);
  \draw (740.5bp,33.0bp) node {1};
  \draw [] (a7_9) ..controls (224.0bp,48.773bp) and (224.0bp,16.821bp)  .. (a8_18);
  \draw (227.5bp,33.0bp) node {2};
  \draw [] (a7_14) ..controls (367.0bp,48.773bp) and (367.0bp,16.821bp)  .. (a8_29);
  \draw (370.5bp,33.0bp) node {1};
  \draw [] (a7_27) ..controls (744.0bp,48.773bp) and (744.0bp,16.821bp)  .. (a8_58);
  \draw (747.5bp,33.0bp) node {2};
  \draw [] (a6_11) ..controls (606.0bp,106.77bp) and (606.0bp,74.821bp)  .. (a7_22);
  \draw (609.5bp,91.0bp) node {1};
  \draw [] (a7_4) ..controls (81.0bp,48.773bp) and (81.0bp,16.821bp)  .. (a8_7);
  \draw (84.5bp,33.0bp) node {2};
  \draw [] (a7_18) ..controls (507.66bp,48.773bp) and (509.37bp,16.821bp)  .. (a8_40);
  \draw (511.5bp,33.0bp) node {2};
  \draw [] (a6_4) ..controls (166.0bp,106.77bp) and (166.0bp,74.821bp)  .. (a7_7);
  \draw (169.5bp,91.0bp) node {2};
  \draw [] (a4_1) ..controls (139.46bp,225.28bp) and (89.067bp,189.1bp)  .. (a5_1);
  \draw (127.5bp,207.0bp) node {1};
  \draw [] (a4_2) ..controls (317.61bp,223.38bp) and (302.3bp,190.42bp)  .. (a5_3);
  \draw (317.5bp,207.0bp) node {1};
  \draw [] (a6_3) ..controls (137.47bp,107.47bp) and (124.72bp,75.02bp)  .. (a7_5);
  \draw (137.5bp,91.0bp) node {2};
  \draw [] (a6_7) ..controls (359.52bp,112.31bp) and (357.91bp,104.62bp)  .. (357.0bp,98.0bp) .. controls (355.34bp,85.917bp) and (354.44bp,71.351bp)  .. (a7_13);
  \draw (360.5bp,91.0bp) node {2};
  \draw [] (a7_11) ..controls (283.53bp,48.773bp) and (287.52bp,16.821bp)  .. (a8_23);
  \draw (290.5bp,33.0bp) node {1};
  \draw [] (a6_9) ..controls (507.0bp,106.77bp) and (507.0bp,74.821bp)  .. (a7_18);
  \draw (510.5bp,91.0bp) node {2};
  \draw [] (a7_26) ..controls (705.0bp,48.773bp) and (705.0bp,16.821bp)  .. (a8_55);
  \draw (708.5bp,33.0bp) node {2};
  \draw [] (a6_11) ..controls (603.84bp,112.63bp) and (601.18bp,104.73bp)  .. (599.0bp,98.0bp) .. controls (595.08bp,85.881bp) and (590.57bp,71.335bp)  .. (a7_21);
  \draw (602.5bp,91.0bp) node {1};
  \draw [] (a6_6) ..controls (293.72bp,106.77bp) and (285.17bp,74.821bp)  .. (a7_11);
  \draw (295.5bp,91.0bp) node {1};
  \draw [] (a6_2) ..controls (73.519bp,112.31bp) and (71.911bp,104.62bp)  .. (71.0bp,98.0bp) .. controls (69.336bp,85.917bp) and (68.441bp,71.351bp)  .. (a7_3);
  \draw (74.5bp,91.0bp) node {1};
  \draw [] (a7_34) ..controls (956.69bp,48.773bp) and (953.27bp,16.821bp)  .. (a8_74);
  \draw (958.5bp,33.0bp) node {1};
  \draw [] (a4_1) ..controls (153.0bp,222.77bp) and (153.0bp,190.82bp)  .. (a5_2);
  \draw (156.5bp,207.0bp) node {2};
  \draw [] (a7_29) ..controls (798.42bp,54.707bp) and (801.28bp,46.859bp)  .. (803.0bp,40.0bp) .. controls (806.0bp,28.04bp) and (807.97bp,13.405bp)  .. (a8_63);
  \draw (808.5bp,33.0bp) node {1};
  \draw [] (a7_35) ..controls (987.79bp,54.825bp) and (992.28bp,47.068bp)  .. (995.0bp,40.0bp) .. controls (999.5bp,28.288bp) and (1002.5bp,13.514bp)  .. (a8_78);
  \draw (1002.5bp,33.0bp) node {1};
  \draw [] (a7_31) ..controls (874.0bp,48.773bp) and (874.0bp,16.821bp)  .. (a8_68);
  \draw (877.5bp,33.0bp) node {2};
  \draw [] (a7_26) ..controls (707.42bp,54.707bp) and (710.28bp,46.859bp)  .. (712.0bp,40.0bp) .. controls (715.0bp,28.04bp) and (716.97bp,13.405bp)  .. (a8_56);
  \draw (717.5bp,33.0bp) node {1};
  \draw [] (a6_16) ..controls (982.66bp,107.3bp) and (1013.1bp,73.821bp)  .. (a7_36);
  \draw (1006.5bp,91.0bp) node {1};
  \draw [] (a7_15) ..controls (403.58bp,54.707bp) and (400.72bp,46.859bp)  .. (399.0bp,40.0bp) .. controls (396.0bp,28.04bp) and (394.03bp,13.405bp)  .. (a8_31);
  \draw (402.5bp,33.0bp) node {1};
  \draw [] (a1_1) ..controls (482.0bp,396.77bp) and (482.0bp,364.82bp)  .. (a2_1);
  \draw (485.5bp,381.0bp) node {~~3};
  \draw [] (a7_14) ..controls (369.42bp,54.707bp) and (372.28bp,46.859bp)  .. (374.0bp,40.0bp) .. controls (377.0bp,28.04bp) and (378.97bp,13.405bp)  .. (a8_30);
  \draw (380.5bp,33.0bp) node {2};
  \draw [] (a4_2) ..controls (330.88bp,223.38bp) and (353.25bp,190.42bp)  .. (a5_4);
  \draw (350.5bp,207.0bp) node {2};
  \draw [] (a7_1) ..controls (16.0bp,48.773bp) and (16.0bp,16.821bp)  .. (a8_2);
  \draw (19.5bp,33.0bp) node {2};
  \draw [] (a7_7) ..controls (167.31bp,48.773bp) and (170.73bp,16.821bp)  .. (a8_14);
  \draw (173.5bp,33.0bp) node {1};
  \draw [] (a7_23) ..controls (627.0bp,48.773bp) and (627.0bp,16.821bp)  .. (a8_49);
  \draw (630.5bp,33.0bp) node {1};
  \draw [] (a7_8) ..controls (187.42bp,54.707bp) and (190.28bp,46.859bp)  .. (192.0bp,40.0bp) .. controls (195.0bp,28.04bp) and (196.97bp,13.405bp)  .. (a8_16);
  \draw (198.5bp,33.0bp) node {1};
  \draw [] (a5_3) ..controls (297.0bp,164.77bp) and (297.0bp,132.82bp)  .. (a6_6);
  \draw (300.5bp,149.0bp) node {1};
  \draw [] (a7_30) ..controls (832.41bp,54.777bp) and (836.47bp,46.985bp)  .. (839.0bp,40.0bp) .. controls (843.28bp,28.204bp) and (846.37bp,13.477bp)  .. (a8_66);
  \draw (846.5bp,33.0bp) node {1};
  \draw [] (a6_15) ..controls (889.09bp,107.47bp) and (878.08bp,75.02bp)  .. (a7_31);
  \draw (888.5bp,91.0bp) node {1};
  \draw [] (a7_19) ..controls (536.0bp,48.773bp) and (536.0bp,16.821bp)  .. (a8_42);
  \draw (539.5bp,33.0bp) node {1};
  \draw [] (a0_1) ..controls (482.0bp,454.77bp) and (482.0bp,422.82bp)  .. (a1_1);
  \draw (485.5bp,439.0bp) node {~~3};
  \draw [] (a7_15) ..controls (406.0bp,48.773bp) and (406.0bp,16.821bp)  .. (a8_32);
  \draw (409.5bp,33.0bp) node {1};
  \draw [] (a4_3) ..controls (585.39bp,223.38bp) and (600.7bp,190.42bp)  .. (a5_6);
  \draw (599.5bp,207.0bp) node {1};
  \draw [] (a5_1) ..controls (66.909bp,165.38bp) and (43.955bp,132.42bp)  .. (a6_1);
  \draw (64.5bp,149.0bp) node {1};
  \draw [] (a7_36) ..controls (1021.7bp,48.773bp) and (1018.3bp,16.821bp)  .. (a8_79);
  \draw (1023.5bp,33.0bp) node {1};
  \draw [] (a5_7) ..controls (785.06bp,164.77bp) and (793.04bp,132.82bp)  .. (a6_14);
  \draw (793.5bp,149.0bp) node {1};
  \draw [] (a6_12) ..controls (673.53bp,106.77bp) and (677.52bp,74.821bp)  .. (a7_25);
  \draw (679.5bp,91.0bp) node {1};
  \draw [] (a6_10) ..controls (543.53bp,106.77bp) and (547.52bp,74.821bp)  .. (a7_20);
  \draw (549.5bp,91.0bp) node {2};
  \draw [] (a6_15) ..controls (893.0bp,106.77bp) and (893.0bp,74.821bp)  .. (a7_32);
  \draw (896.5bp,91.0bp) node {1};
  \draw [] (a7_25) ..controls (681.42bp,54.707bp) and (684.28bp,46.859bp)  .. (686.0bp,40.0bp) .. controls (689.0bp,28.04bp) and (690.97bp,13.405bp)  .. (a8_54);
  \draw (691.5bp,33.0bp) node {1};
  \draw [] (a5_8) ..controls (906.71bp,167.28bp) and (957.75bp,131.1bp)  .. (a6_16);
  \draw (944.5bp,149.0bp) node {1};
  \draw [] (a5_1) ..controls (75.0bp,164.77bp) and (75.0bp,132.82bp)  .. (a6_2);
  \draw (78.5bp,149.0bp) node {2};
  \draw [] (a7_12) ..controls (307.52bp,54.312bp) and (305.91bp,46.618bp)  .. (305.0bp,40.0bp) .. controls (303.34bp,27.917bp) and (302.44bp,13.351bp)  .. (a8_24);
  \draw (308.5bp,33.0bp) node {1};
  \draw [] (a6_3) ..controls (142.0bp,106.77bp) and (142.0bp,74.821bp)  .. (a7_6);
  \draw (145.5bp,91.0bp) node {1};
  \draw [] (a7_24) ..controls (666.0bp,48.773bp) and (666.0bp,16.821bp)  .. (a8_52);
  \draw (669.5bp,33.0bp) node {1};
\end{tikzpicture}
}
\end{center}
\caption{The labeled orbit tree of the Mealy automaton of Figure~\ref{fig-conn} (up to level~7).}
\label{fig-otree}
\end{figure}

{\bf From now on we will assume that all the orbit trees are labeled.}

Let \(\mot{u}\) be a (possibly infinite) word over \(Q\). The \emph{path
of}~\(\mot{u}\) in the orbit tree~\(\otree[A]\) is the unique
initial path going from the root through the connected
components of the prefixes of~\(\mot{u}\); \(\mot{u}\) can be called
\emph{a representative} of this initial path (or of the orbit of~\(\mot{u}\)
under~\(\langle\dz(\aut{A})\rangle\) representing the endpoint of this path);
we can say equivalently that this path is \emph{represented} by~\(\mot{u}\).

\begin{definition}
Let \(e\) and~\(f\) be two edges in the orbit tree~\(\otree[A]\).
We say that~\(e\) \emph{is liftable to}~\(f\) if each word of~\(\bot(e)\)
admits some word of~\(\bot(f)\) as a suffix.
\end{definition}

Obviously if \(e\) is liftable to~\(f\), then \(f\) is closer to the root of the orbit tree.
The fact that an edge is liftable to another one reflects a deeper relation stated below.
The following lemma is one of the key observations that we use many times later in the paper.

\begin{lemma}\label{lem-liftable}
Let \(e\) and~\(f\) be two edges in the orbit tree~\(\otree[A]\). If \(e\) is
liftable to~\(f\), then the label of~\(e\) is less than or equal to the label
of~\(f\).
\end{lemma}

\begin{proof}
Since~\(e\) is liftable to~\(f\), each word in~\(\bot(e)\) has a form
\(\mot{vu}x\) for some~\(\mot{u}x\in\bot(f)\). Suppose that \(\mot{vu}x\)
and~\(\mot{vu}y\) are in the same connected component:
there exists \(\mot{s}\in\Sigma^*\) such that~\(\delta_{\mot{s}}\) moves
\(\mot{vu}x\) to \(\mot{vu}y\). In this case,
\(\rho_{\mot{v}}(\mot{s})\) moves \(\mot{u}x\) to \(\mot{u}y\).
Thus, the number of children of~\(\mot{vu}\) in the connected component
of~\(\mot{vu}x\) (which is equal to the label of~\(e\)) is less than or equal to
the number of children of~\(\mot{u}\) in the connected component of
$\mot{u}x$ (which is equal to the label of~\(f\)).
\end{proof}

This notion can be generalized to paths:
\begin{definition}
Let \(\mot{e}=(e_i)_{i\in I}\)
and~\(\mot{f}=(f_i)_{i\in I}\) be two paths of the same (possibly
infinite) length in the orbit tree~\(\otree[A]\). The path \(\mot{e}\) \emph{is liftable to} the path
\(\mot{f}\) if, for any \(i\in I\), the edge~\(e_i\) is liftable to the edge \(f_i\).
\end{definition}

As each word~\(\mot{u}\in Q^*\) is a state in a connected component
of~\(\aut{A}^{|\mot{u}|}\), we can notice the following fact which is
crucial for all our forthcoming proofs.

\begin{lemma}\label{lem-lift-up}
Let \(\mot{e}\) be a path at level~\(k\) in the orbit tree~\(\otree[A]\).
Then, for any~\(\ell<k\), \(\mot{e}\) is liftable to some path
at level~\(\ell\). In particular, \(\mot{e}\) is liftable to
some initial path.
\end{lemma}

\section{Main result}
\label{sec-main_result}
We study here the case where \(\aut{A}\) is a connected
invertible-reversible 3-state Mealy automaton, which means that the
orbit tree of its dual has a unique edge adjacent to the root, labeled by~3.
We prove that if \(\aut{A}\) generates an
infinite group, then the orbit tree~\(\otree[A]\) admits a (necessarily unique)
branch without edges labeled by~1, more precisely a branch of label
either~\(3^{\omega}\) or~\(3^n2^{\omega}\), where $i^\omega$ denotes an infinite word whose each letter is $i$. An element of infinite order
will then be constructed using this branch.

\medbreak The restriction to connected Mealy automata is discussed in Remark~\ref{remarkDisconnected}.

\subsection{General structure of the orbit tree}

From Lemma~\ref{lem-lift-up}, we obtain the following result
on the connection degree of a Mealy automaton (note that it does not depend on the number of states in the Mealy automaton):
\begin{proposition}\label{prop-lin3}
If for some~\(n\), a connected component of~\(\aut{A}^n\) does not
split up, then the connection degree of~\(\aut{A}\) is at least \(n+1\).
\end{proposition}

\begin{proof}
Suppose that~\(\aut{A}\) has~\(m\) states. If an edge at level~\(n\) in the orbit tree~\(\otree[A]\)
has label~$m$, by Lemma~\ref{lem-lift-up}, it is liftable to
some edge at any level above \(n\), and by Lemma~\ref{lem-liftable} this
edge is labeled by~$m$. Now, by going from top to bottom, we can
conclude that there is only one edge at each level above \(n+1\) in the
orbit tree.
\end{proof}

Now we restrict our attention again on the case of 3-state automata, unless specified otherwise. If the connection degree of such an automaton~\(\aut{A}\) is infinite, it has
been proved in~\cite[Proposition~14]{Kli13} that the generated semigroup is free,
freely generated by the states of the Mealy automaton.
As a consequence \(\aut{A}\) cannot generate an infinite
Burnside group in this case. {\bf So from now on, we assume~\(0<\cd{A}<\infty\)}.

We know now that the orbit tree~\(\otree[A]\) has a prefix
linear part until level \(\cd{A}\) and that below this level, all the
vertices split up. We denote by~\(\cdv{A}\) the highest vertex to
split up (\emph{i.e.} the only vertex at level~\(\cd{A}\)).

\begin{definition}
Let \(\mot{i}\) be a (possibly infinite) word over an alphabet~\(F\)
and let~\(j\in F\). A~\(j\)-block~\(\mot{j}\)
of~\(\mot{i}\) is a maximal factor of~\(\mot{i}\) in~\(j^*\cup\{j^\omega\}\), that is,
\(\mot{i}=\mot{k}\mot{j}\mot{l}\) holds, where the last letter
of~\(\mot{k}\) and the first letter of~\(\mot{l}\), if not empty, are not~\(j\).
\end{definition}

\begin{lemma}\label{lem-2-blocks}
If the lengths of the 2-blocks in the orbit tree \(\otree[A]\) are not
bounded, then \(\otree[A]\) admits a branch labeled by
\(3^{\cd{A}}2^{\omega}\). If the lengths of the 2-blocks are bounded with
supremum~\(N\), then \(\otree[A]\) admits an initial path labeled
by~\(3^{\cd{A}}2^{N}\) (and none labeled by~\(3^{\cd{A}}2^{N+1}\)).
\end{lemma}

\begin{proof}
As there is at most one path starting at~\(\cdv{A}\) with a maximal
prefix in~\(2^\omega\) because the stateset has size~3,
Lemma~\ref{lem-lift-up} leads to the conclusion.
\end{proof}

By Proposition~\ref{prop-lin3}, below the connection degree of~\(\aut{A}\),
no edge can be labeled by~3. On the other hand, the case
when all edges are labeled by~1
turns out to be irrelevant for generating infinite Burnside groups,
according to the following proposition (that holds for automata with arbitrary number of states).

\begin{proposition}\label{prop-total-split-up}
If all edges coming down from the only connected component
at vertex~\(\cdv{A}\) are labeled by~1,
the group generated by~\(\aut{A}\) is finite.
\end{proposition}

\begin{proof}
By Lemmas~\ref{lem-liftable} and~\ref{lem-lift-up}, below
level \(\cd{A}\), all edges are then labeled by~$1$, which means that the
connected components of the powers of~\(\aut{A}\) have bounded size,
and the group~\(\pres{\aut{A}}\) is finite according to Proposition~\ref{prop-bounded-cc}.
\end{proof}

{\bf From now on, we assume that the only connected component at vertex
\(\cdv{A}\) splits up in two connected components.}

\subsection{Reduction edge and orbital words}
Let us recall our framework: \(\aut{A}\) is a connected 3-state
invertible-reversible Mealy automaton with~\(0<\cd{A}<\infty\),
such that the only connected component at vertex~\(\cdv{A}\)
splits up in two connected components.
The point is now to put the emphasize on the larger of the two.

\begin{definition}
At level \(\cd{A}\)
there are two edges, one labeled by~1 and a second one labeled by~2.
We call by the \emph{reduction edge} this last edge and denote it
by~\(\redEdge\), and denote by~\(\two{A}\)
its terminal vertex, and call it the \emph{reduction orbit}.
\end{definition}

\begin{lemma}\label{lem-2extends}
Each vertex below the
level~\(\cd{A}+1\) is the initial vertex of either one edge which is liftable to
the reduction edge and one edge which
is not, or two edges which are liftable to the reduction edge and one which is not.
\end{lemma}

\begin{definition}\label{defi-restricted-path}
An \emph{\(\redEdge\)-liftable path} is an initial (possibly infinite) path in the orbit tree whose
edges below the level \(\cd{A}\) are all liftable to the reduction edge $\redEdge$
of~\(\otree[A]\).
\end{definition}

\begin{definition}\label{def-orbital}
A word over~\(Q\) is said to be~\emph{orbital}
if it is a representative of an \(\redEdge\)-liftable path or, equivalently,
if all its length~\(\cd{A}\) factors belong to the reduction orbit~\(\two{A}\).
Denote then by~\(\htree\) the set of all finite orbital words.
Being prefix-closed, \(\htree\) can be
seen either as a set of words, or as a tree.
\end{definition}

Note that~\(\htree\) is a strictly (\(\cd{A}\)+1)-testable language~\cite{local-test}.
In particular, for any word~$\mot{u}\in\htree$ of length at least~$\cd{A}$,
the set of words~$\mot{v}$ such that~$\mot{uv}$ belongs to~$\htree$ depends
only on the length~\(\cd{A}\) suffix of~\(\mot{u}\). A simple consequence is that
\(\htree\) viewed as a set of infinite words is factor-closed.
Further, up to level~$\cd{A}$, the tree~$\htree$ coincides with~$Q^*$
and each word of length at least~$\cd{A}$ is a prefix of exactly
two other words in~$\htree$.

\medbreak We now state several technical results
about the tree~\(\htree\) of orbital words which will be of use
in the final subsection.

\begin{lemma}\label{lem-successor}
For any two words~\(\mot{u},\mot{v}\) in~\(\htree\)
and any integer~\(n\), there exists a word~\(\mot{r}\in Q^n\)
satisfying~\(\mot{u}\mot{r}\in\htree\) and~\(\mot{v}\mot{r}\in\htree\).
\end{lemma}

\begin{proof}
For any word~\(\mot{w}\in\htree\), there exist (at least) two different states~\(q_1,q_2\in{Q}\)
satisfying~\(\mot{w}q_1,\mot{w}q_2\in\htree\).
As the stateset~\(Q\) has size~3, there is at least one common state~\(r_1\in Q\)
satisfying~\(\mot{u}r_1\in\htree\) and~\(\mot{v}r_1\in\htree\).
The result is then obtained by recursion.
\end{proof}

\begin{lemma}
\label{lem:full_orbit}
Let \(x\in Q\) and \(\mot{u}\in Q^{\cd{A}}\).
The set of length~\(\cd{A}\) suffixes of all words in the connected component of~\(x\mot{u}\) is the whole~\(Q^{\cd{A}}\).
\end{lemma}

\begin{proof}
Since \(\aut{A}\) is reversible and \(\aut{A}^{\cd{A}}\) connected,
for each \(\mot{v}\in{Q}^{\cd{A}}\),
there exists \(\mot{s}\in\Sigma^*\) such that
\(\delta_{\mot{s}}(\mot{u})=\mot{v}\). By the invertibility of \(\aut{A}\),
\(\mot{t}=\rho_{x}^{-1}(\mot{s})\) is well defined and we have
\[\delta_{\mot{t}}(x\mot{u})=\delta_{\mot{t}}(x)\delta_{\rho_x(\mot{t})}(\mot{u})
=\delta_{\mot{t}}(x)\delta_{\mot{s}}(\mot{u})= \delta_{\mot{t}}(x)\mot{v}\:.\]\end{proof}

\begin{proposition}\label{prop-subtree}
For any orbital word \(\mot{u}\), there are infinitely many edges
in~\(\vhhtree[u]\) labeled by each state.
\end{proposition}

\begin{proof}
By contradiction. Denote the stateset \(Q=\{x,y,z\}\) and let
\(\mot{u}\) be an orbital word (of length at least~\(\cd{A}\))
such that no edge of~\(\vhhtree[u]\) is labeled by~\(z\).

As each word of~\(\htree\) can be extended in~\(\htree\) by two
different states, \(x\) and~\(y\) are in~\(\vhhtree[u]\). By
recursion, \(\{x,y\}^*\) is a subset of~\(\vhhtree[u]\) and, as
\(\htree\) is suffix-closed, \(\{x,y\}^*\subseteq\htree\).
Let~\(\mot{v}\in \{x,y\}^{\cd{A}-1}\): \(x\mot{v}\) and~\(y\mot{v}\) are
elements of~\(\htree\) and~\(x\mot{v}x\), \(x\mot{v}y\),
\(y\mot{v}x\), and~\(y\mot{v}y\) are in~\(\two{\aut{A}}\). Hence
\(x\mot{v}z\) and~\(y\mot{v}z\) have length~\(\cd{A}+1\) and are not
in~\(\two{\aut{A}}\). In this case both these words must belong to the other connected component of size~\(3^{\cd{A}}\).

By Lemma~\ref{lem:full_orbit}, the connected component of~\(x\mot{v}z\) will
have at least~\(3^{\cd{A}}\) words with different suffixes starting from position~\(2\).
Therefore, if we assume that~\(y\mot{v}z\) is also in this component, we must
have at least~\(3^{\cd{A}}+1\) words in there. Contradiction.
\end{proof}

\subsection{Cyclically orbital words and elements of infinite order}
In this subsection, in the case when the lengths of the
2-blocks in the orbit tree are bounded, we exhibit a family of words whose induced actions have finite
bounded orders. Then we prove that each word admits a bounded power which
induces the same action as some word in this family.

\begin{definition}
A word over~\(Q\) is said to be~\emph{cyclically orbital} if each of its powers is orbital.
In other words, a word is cyclically orbital
if viewed as a cyclic word it is orbital.
\end{definition}

Note that the existence of such cyclically orbital words
is ensured by the straightforward fact that any orbital word
of length~\(\cd{A}\times(1+\#Q^{\cd{A}})\) admits a cyclically orbital factor.

\medbreak To prove the first result about cyclically orbital words we will need the following proposition.

\begin{proposition}\label{prop-2-blocks}
If the lengths of the 2-blocks in~\(\otree[A]\) are bounded with supremum~\(
N\), then each edge from level \(\cd{A}+N\) or below in an \(\redEdge\)-liftable path is followed by three edges.
\end{proposition}

\begin{proof} From Lemma~\ref{lem-2-blocks}, there is
an initial path
labeled by~\(3^{\cd{A}}2^N\) (and none labeled
by~\(3^{\cd{A}}2^{N+1}\)). As the stateset has size~3, this path is
unique, call it \(\mot{e}\):
\(\bot(\mot{e})\) is the initial vertex of 3 edges,
two of which are liftable to the reduction edge by
Lemma~\ref{lem-2extends}.

We now consider the subtree~\({\rest}\) of~\(\otree[A]\)
consisting of all \(\redEdge\)-liftable paths.
By Lemma~\ref{lem-2extends}, each path of~\({\rest}\) is
either a prefix of~\(\mot{e}\) or prefixed with~\(\mot{e}\).
Suppose that one of
the edges of~\({\rest}\) below level~\(|\mot{e}|+1\) is labeled by~2 and
consider an initial branch~\(\mot{f}\) in~\({\rest}\) which minimizes the
length of the 1-block from~\(\bot(\mot{e})\): its label has prefix~\(3^{\cd{A}}2^N1^k2\)
for some~\(k>0\).
Consider the (non-initial) path in $\otree[A]$ obtained from \(\mot{f}\) by erasing the first edge:
by Lemma~\ref{lem-lift-up}, it is liftable to an initial path,
say~\(\mot{g}\). As \(\mot{f}\) is an \(\redEdge\)-liftable path, so is~\(\mot{g}\),
hence~\(\mot{g}\) coincides necessarily with~\(\mot{e}\) until level~\(\cd{A}+N\);
so the label of~\(\mot{g}\) has prefix~\(3^{\cd{A}}2^N\).
By Lemmas~\ref{lem-liftable} and~\ref{lem-lift-up},
it has also a prefix whose label is greater than or equal
to (coordinatewise) \(3^{\cd{A}-1}2^N1^k2\). Hence the
label of~\(\mot{g}\) has a prefix greater than or equal to~\(3^{\cd{A}}2^N1^{k-1}2\),
which is in contradiction with the choice of~\(\mot{f}\).
\end{proof}

\begin{proposition}\label{prop-const-order}
If the lengths of the 2-blocks in~\(\otree[A]\) are bounded, any
cyclically orbital word induces an action of finite order, bounded by a
constant not depending on the word.
\end{proposition}

\begin{proof}
Let \(\mot{u}\) be a cyclically orbital word and~\(n\) be an integer such
that~\(|\mot{u}^n|>\cd{A}\). By the definition,
\(\mot{u}^n\) is a representative of some \(\redEdge\)-liftable path. So, by
Lemma~\ref{lem-2-blocks}, the label of the path of
\(\mot{u}^{\omega}\) is ultimately~1 and, by
Proposition~\ref{prop-bounded-cc}, the action induced by~\(\mot{u}\)
has finite order, bounded by a constant which depends on \(\cd{A}\)
(more precisely on how many ways one can choose the outputs for a
Mealy automaton with the same structure as~\(\bot(\mot{e})\) where
\(\mot{e}\) is the path of~\(\otree[A]\) defined in the proof of
Proposition~\ref{prop-2-blocks}).
\end{proof}

\begin{proposition}\label{prop-equiv-adm}
If the lengths of the 2-blocks in~\(\otree[A]\) are bounded, any
non-empty word over~\(Q\) admits a non-empty bounded power which is equivalent to
a cyclically orbital word.
\end{proposition}

\begin{proof}
Let~$\cw=\cd{A}$ be the connection degree of~$\aut{A}$.
For a (possibly infinite) word~$\mot{w}\in\htree$, let~\(\fact_\cw(\mot{w})\) denote the (finite) set of its length~\(\cw\) factors.

Consider an infinite word~\(\mot{u}\in\htree\), that we assume to be \emph{maximal} in the sense that there is no other
infinite word~\(\mot{u'}\in\htree\) satisfying~\(\fact_\cw(\mot{u})\varsubsetneq\fact_\cw(\mot{u'})\).
Let fix a finite prefix~\(\mot{v}\) of~\(\mot{u}\) satisfying~\(\fact_{\cw}(\mot{u})=\fact_{\cw}(\mot{v})\).
From the maximality assumption on~\(\mot{u}\), we deduce that each~\(\mot{w}\in\vhhtree\) (where $\vhhtree$, viewed as a set of words, consists of words $\mot{w}$ over $Q$ such that $\mot{vw}\in\htree$) satisfies~\(\fact_{\cw}(\mot{w})\subseteq\fact_{\cw}(\mot{v})\). We will refer to this property as Property~(\({\natural}\)).

Let \(a_1a_2\dots a_n\in Q^+\).
By Proposition~\ref{prop-subtree}, \(a_1\) appears infinitely many often in the tree~\(\vhhtree\).
Therefore there is some word~\(\mot{u_0}\) satisfying~\(\mot{u_0}a_1\in\vhhtree\).
The goal is to build a word~\(\mot{u_1}\) satisfying~\(\mot{u_0}a_1\mot{u_1}a_2\in\vhhtree\).
As in Figure~\ref{fig-constr}, choose some word~\(\mot{v_0}\) satisfying~\(\mot{v_0}a_2\in\vhhtree\).
By Lemma~\ref{lem-successor} and Property~(\(\natural\)), there exists a word \(\mot{r}\in\fact_{\cw}(\mot{v})\) such that both \(\mot{u_0}a_1\mot{r}\in\vhhtree\) and~\(\mot{v_0}\mot{r}\in\vhhtree\) hold.
Let~\(\mot{v}_{\mot{r}}\) be the shortest word such that~\(\mot{v}_{\mot{r}}\mot{r}\) is a prefix of~\(\mot{v}\) and
let~\(\mot{u_1'}\) be the word satisfying~\(\mot{v}_{\mot{r}}\mot{u_1'}=\mot{v}\mot{v_0}\). Then \(\mot{u_1'}\) is a cyclically orbital word, since its length~\(\cw\) prefix~\(\mot{r}\) satisfies~\(\mot{v}_{\mot{r}}\mot{u_1'}\mot{r}\in\htree\).
\begin{figure}[h]
\centering
\scalebox{1}{\begin{tikzpicture}[yscale=.6]
\coordinate (A) at (0,0); 			\coordinate (Aw) at (1,0);
\coordinate (B) at (0,-1);
\coordinate (C) at (0,-2);
\coordinate (D) at (0,-3);			\coordinate (Dw) at (1,-3);			\coordinate (Dww) at (4,-3);
\coordinate (E) at (-4,-9.7);
\coordinate (F) at (4,-9.7);										
\coordinate (G) at (.7,-5);
\coordinate (H) at (1.5,-6);
\coordinate (I) at (1.4,-6.2);
\coordinate (J) at (1.4,-7.2);
\coordinate (J1) at (1.4,-8.2);
\coordinate (J2) at (.8,-9.7);
\coordinate (K) at (-.6,-4.5);
\coordinate (L) at (-.5,-4.7);
\begin{scope}[color=lightgray]
{\pgfmathsetseed{1234}
\draw[randomPath,line width=5pt] (B) -- +(-.25,-.33) -- +(.05,-.7) -- +(0,-1);}
{\pgfmathsetseed{4321}
\draw[randomPath,line width=5pt] (C) -- node[left]{{\color{gray}\(\mot{u_1'}\)}} (D);
\draw[randomPath,line width=5pt] (D) -- (K);}
{\pgfmathsetseed{1234}
\draw[randomPath,line width=5pt] (I) -- +(-.25,-.33) -- +(.05,-.7) -- +(0,-1);}
{\pgfmathsetseed{4321}
\draw[randomPath,line width=5pt] (J) -- node[left]{{\color{gray}\(\mot{u_1'}\)}} (J1);
\draw[randomPath,line width=5pt] (J1) -- (J2);}
\end{scope}
{\pgfmathsetseed{1234}
\draw[randomPath] (B) -- +(-.25,-.33) -- node[left]{\(\mot{r}\)}
     +(.05,-.7) -- +(0,-1);}
\draw[randomPath,densely dotted] (A) -- node[right] {\(\mot{v}_\mot{r}\)} (B);
{\pgfmathsetseed{4321}
\draw[randomPath,thin] (C) -- (D);}
\path  	(D) edge[dotted,bend left] (F)
		(D) edge[dotted,bend right] (E);
\draw[randomPath,densely dotted] (D) -- node[left] {\(\mot{v_0}\)} (K);
\draw[randomPath2,densely dotted,in=0] (D) -- node[right] {\(\mot{u_0}\)} (H);
\draw (H) -- node[right]{\(a_1\)} (I);
{\pgfmathsetseed{1234}
\draw[randomPath] (I) -- +(-.25,-.33) -- node[left]{\(\mot{r}\)}
     +(.05,-.7) -- +(0,-1);}
\draw[thick] (K) -- node[right]{\(a_2\)} (L);
\draw[thick] (J2) -- node[right]{\(a_2\)} +(.1,-.2);
{\pgfmathsetseed{1234}
\draw[randomPath] (K) -- +(-.25,-.33) -- node[left]{\(\mot{r}\)}
     +(.05,-.7) -- +(0,-1);}
\draw[decorate,decoration=brace] (Aw) -- (Dw) node[right,midway] {\(~\mot{v}\)};
\draw[white] (Dww.east) -- (F.east) node[midway] {\color{black}\(\vhhtree\)};
\end{tikzpicture}}
\caption{Proof of~Proposition~\ref{prop-equiv-adm}: building a word~\(\mot{u_1'}\) satisfying
\(\mot{u_0}a_1\mot{u_1'}a_2\in\vhhtree[\mot{v}]\).}\label{fig-constr}
\end{figure}

By Proposition~\ref{prop-const-order}, \(\mot{u_1'}\) has finite
order~\(q\). So we set~\(\mot{u_1}=\mot{u_1'}^q\)
and keep~\(a_1\mot{u_1}a_2\in\vhhtree[\mot{v}\mot{u_0}]\).

The same method produces words~\((\mot{u_i})_{1\leq\mot{i}\leq n}\) of length at least~\(\cw\)
that induce the trivial action and such that the word~\(\mot{w^{(0)}}=a_1 \mot{u_1}a_2\cdots a_{n-1}\mot{u_{n-1}}a_n\mot{u_n}\) satisfies~\(\mot{w^{(0)}}a_1\in\vhhtree[\mot{v}\mot{u_0}]\) and induces the same action as~\(a_1\cdots a_n\). For~\(i\geq0\), we analogously define a word~\(\mot{w^{(i+1)}}\) inducing the same action as~\(a_1\cdots a_n\) such that \(\mot{w^{(i+1)}}a_1\in\vhhtree[\mot{v}\mot{u_0}\mot{w^{(0)}}\cdots\mot{w^{(i)}}]\). Eventually, there exist~\(\mot{i}<\mot{j}\) such that~\(\mot{u_1^{(i)}}\) and~\(\mot{u_1^{(j+1)}}\) have the same prefix of length~\(\cw\), hence \(\mot{w^{(i)}}\cdots\mot{w^{(j)}}\) is cyclically orbital and induces the same action as~\((a_1\cdots
a_n)^{\mot{j}-\mot{i}}\). Note that~\(\mot{j}-\mot{i}\) is bounded by
a constant depending on \(\cd{A}\) and~\(\#Q\).
\end{proof}

\begin{corollary}\label{cor-branch}
If \(\aut{A}\) generates an infinite group, the orbit tree of its dual
admits an \(\redEdge\)-liftable branch labeled either by~\(3^{\omega}\) or by
\(3^n2^{\omega}\) for some~\(n\).
\end{corollary}

\begin{proof}
If the lengths of the 2-blocks are bounded, any non-empty word has a
non-empty bounded power which is equivalent to
a cyclically orbital word by Proposition~\ref{prop-equiv-adm}, so the order
of the action it induces is
finite and bounded by a constant from
Propositions~\ref{prop-const-order} and~\ref{prop-equiv-adm}.

It follows from Zelmanov's solution to the restricted Burnside
problem~\cite{Vl,Ze1,Ze2} that the group \(\pres{\aut{A}}\) is finite, which
contradicts the hypothesis. Proposition~\ref{prop-2-blocks} leads to
the conclusion.
\end{proof}

\begin{nremark}\label{remarkDisconnected} A disconnected 3-state invertible-reversible Mealy
automaton have either three connected components -- each of size~1 -- and the generated group is finite (see~Proposition~\ref{prop-total-split-up}),
or two connected components -- of size~1 and~2 -- and the situation is then more subtle because
both connected components can generate separately finite groups, but
together an infinite group. Such a disconnected Mealy automaton might not satisfy Corollary~\ref{cor-branch}.
For example, the Mealy automaton in Figure~\ref{fig-disc} illustrates this situation. The group that it generates is infinite since the product of two components of this automaton generates an infinite subgroup~\(\pres{xy,xz}\)~\cite{AKLMP12,Kli13}.
However, if there is an initial branch labeled by $2^\omega$ and the analogue of Corollary~\ref{cor-branch} holds, it would imply that both $y$ and $z$ have infinite order, which is not the case as $\langle y,z\rangle$ is finite.
\end{nremark}

\begin{theorem}
Any infinite group generated by a connected 3-state invertible-reversible Mealy
automaton admits an element of infinite order.
\end{theorem}

\begin{proof}
All cyclically orbital words induce actions of infinite order from
Proposition~\ref{prop-finite} and Corollary~\ref{cor-branch}.
\end{proof}

Finally we can emphasize that the proof of existence of elements of infinite order is constructive and that the effective detection of such elements becomes computable in this framework. For instance, in the group generated by the Mealy automaton in Figure~\ref{fig-conn}, the cyclically orbital word~\(xyz\) has infinite order. At least in this case, neither one of the existing two packages \FR~\cite{FR} and \SK~\cite{SK} for~\GAP~\cite{GAP} system, dedicated to Mealy automata and groups they generate, is able to detect such an element of infinite order.

\bibliographystyle{amsplain}
\bibliography{3states}

\end{document}